\documentclass[a4paper,USenglish,cleveref, autoref, thm-restate]{lipics-v2021}
\pdfoutput=1 

\title{Global Type Inference for Featherweight Generic Java} 

\author{Andreas Stadelmeier}{DHBW Stuttgart, Campus Horb, Germany}{a.stadelmeier@hb.dhbw-stuttgart.de}{}{}

\author{Martin Plümicke}{DHBW Stuttgart, Campus Horb, Germany}{pl@dhbw.de}{}{}

\author{Peter Thiemann}{Universität Freiburg, Institut für Informatik, Germany}{thiemann@informatik.uni-freiburg.de}{}{}

\authorrunning{A. Stadelmeier and M. Plümicke and P. Thiemann} 

\Copyright{Andreas Stadelmeier and Martin Plümicke and Peter Thiemann} 



\category{} 

\relatedversion{} 

\acknowledgements{}

\nolinenumbers 

\EventEditors{John Q. Open and Joan R. Access}
\EventNoEds{2}
\EventLongTitle{42nd Conference on Very Important Topics (CVIT 2016)}
\EventShortTitle{CVIT 2016}
\EventAcronym{CVIT}
\EventYear{2016}
\EventDate{December 24--27, 2016}
\EventLocation{Little Whinging, United Kingdom}
\EventLogo{}
\SeriesVolume{42}
\ArticleNo{23}

\usepackage{cite}
\usepackage [disable]{todonotes} 
\usepackage[utf8]{inputenc}
\usepackage{hyperref}
\usepackage{amsmath}
\usepackage{amssymb}
\usepackage{subcaption}
\usepackage{prftree}
\usepackage{tabularx}
\usepackage{multicol}
\usepackage{xspace}
\usepackage{color,ulem}
\usepackage{listings}
\lstdefinestyle{fgj}{backgroundcolor=\color{lime!20}}
\lstdefinestyle{tfgj}{backgroundcolor=\color{lightgray!20}}

\newcommand\mv[1]{{\tt #1}}
\newcommand{\ol}[1]{\overline{\tt #1}}
\newcommand{\exptype}[2]{\mathtt{#1 \texttt{<} #2 \texttt{>} }}
\newcommand\ddfrac[2]{\frac{\displaystyle #1}{\displaystyle #2}}

\newcommand{\olsub}{\textrm{$<:$}\ }

\newcommand{\sub}\olsub

\newcommand{\set}[1]{\{ #1 \} }

\definecolor{red}{rgb}{1,0,0}
\newcommand{\red}[1]{\textcolor{red}{#1}}

\newcommand{\commentary}[1]{\marginpar[\tiny
  \red{#1}]{\tiny \red{#1}}}

\newcommand\Erase[1]{|#1|}
\newcommand\Angle[1]{\langle#1\rangle}

\newcommand\TFGJ{FGJ-GT\xspace}
\newcommand\FGJGT{FGJ-GT\xspace}

\newcommand\TVX{\mv X}
\newcommand\TVY{\mv Y}
\newcommand\TVZ{\mv Z}
\newcommand\TVW{\mv W}

\newcommand\CL[1]{\textit{Cl}$_{#1}$}
\newcommand\subconstr{\lessdot}
\newcommand\eqconstr{\doteq}
\newcommand\subeq{\mathbin{\texttt{<:}}}
\newcommand\extends{\ensuremath{\triangleleft}}

\newcommand\rulename[1]{\textup{\textrm{(#1)}}}

\newcommand{\environmentvdash}{\Pi;\Delta;\Gamma \vdash}

\newcommand{\fjtypeinference}{\textbf{FJTypeInference}}

\newcommand{\tv}[1]{\mathit{ #1 }}
\newcommand{\unifyGenerics}{\ensuremath{\gamma}}
\newcommand{\fjtype}{\textbf{FJType}}
\newcommand{\unify}{\textbf{Unify}}
\newcommand{\typeMethod}{\textbf{TYPEMethod}}
\newcommand{\typeExpr}{\textbf{TYPEExpr}}
\newcommand{\constraint}{\ensuremath{\mathit{c}}}
\newcommand{\consSet}{C}
\newcommand{\orCons}{\textit{oc}}
\newcommand{\simpleCons}{\textit{sc}}

\newcommand{\typeAssumptionsSymbol}{\ensuremath{\Theta}}
\newcommand{\constraints}{\ensuremath{\mathit{\overline{c}}}}
\newcommand{\itype}[1]{\ensuremath{\mathit{#1}}}
\newcommand{\il}[1]{\ensuremath{\overline{\itype{#1}}}}
\newcommand{\type}[1]{\texttt{#1}}

\newcommand{\mtypeEnvironment}{\ensuremath{\Pi}}
\newcommand{\methodAssumption}{\ensuremath{\mathtt{\lambda}}}
\newcommand{\localVarAssumption}{\ensuremath{\mathtt{\eta}}}
\newcommand{\expandLB}{\textit{expandLB}}

\def\exptypett#1#2{\texttt{#1}\textrm{{\tt <#2}}\textrm{{\tt >}}\xspace}
\def\exp#1#2{#1(\,#2\,)\xspace}

\begin{document}

\maketitle

\begin{abstract}
  Java's type system mostly relies on type checking augmented with
  local type inference to improve programmer convenience.

  We study global type inference for Featherweight Generic Java (FGJ), a
  functional Java core language. Given generic class headers and field
  specifications, our inference algorithm infers all method types if
  classes do not make use of polymorphic recursion.
  The algorithm is constraint-based and improves on prior work in
  several respects. Despite the restricted setting, global type
  inference for FGJ is NP-complete.
\end{abstract}

\section{Introduction}
\label{sec:introduction}

Java is one of the most important programming languages. In 2019, Java
was the second most popular language according to a study
based on GitHub
data.\footnote{\url{https://www.businessinsider.de/international/the-10-most-popular-programming-languages-according-to-github-2018-10/}} Estimates
for the number of Java programmers range between 7.6 and 9 million.\footnote{\url{https://www.zdnet.com/article/programming-languages-python-developers-now-outnumber-java-ones/},
\url{http://infomory.com/numbers/number-of-java-developers/}} Java
has been around since 1995 and progressed through 16 versions.

Swarms of programmers have taken their first steps in Java. Many more
have been introduced to object-oriented programming through Java, as
it is among the first mainstream languages supporting
object-orientation. Java is a class-based language with static single inheritance among
classes, hence it has nominal types with a specified subtyping
hierarchy. Besides classes there are interfaces to characterize common 
traits independent of the inheritance hierarchy. Since version J2SE~5.0,
the Java language supports F-bounded polymorphism in the form of generics.

Java is generally explicitly typed with some amendments introduced in
recent versions. That is, 
variables, fields, method parameters, and method returns must be
adorned with their type. Figure~\ref{fig:intro-example-generic-fj}
contains a simple example with generics.
\begin{figure}[tp]
  \begin{subfigure}[t]{0.49\linewidth}
\begin{lstlisting}[style=fgj]
class Pair<X,Y> {
  X fst;
  Y snd;
  Pair<X,Y>(X fst, Y snd) {
    this.fst=fst;
    this.snd=snd;
  }
  <Z> Pair<Z,Y> setfst(Z fst) {
    return new Pair(fst, this.snd);
  }
  Pair<Y,X> swap() {
    return new Pair(this.snd, this.fst);
  }
}  
\end{lstlisting}
    \caption{Featherweight Generic Java (FGJ)}
    \label{fig:intro-example-generic-fj}
  \end{subfigure}
  ~
  \begin{subfigure}[t]{0.49\linewidth}
\begin{lstlisting}[style=tfgj]
class Pair<X,Y> {
  X fst;
  Y snd;
  Pair(fst, snd) {
    this.fst=fst; 
    this.snd=snd;
  }
  setfst(fst) {
    return new Pair(fst, this.snd);
  }
  swap() {
    return new Pair(this.snd, this.fst);
  }
}  
\end{lstlisting}
    \caption{FGJ with global type inference (\TFGJ)}
    \label{fig:intro-example-generic-jtx}
  \end{subfigure}
  \caption{Example code}
  \label{fig:intro-example-code}
\end{figure}

Our type inference algorithm is able to infer generalized method types
as shown in figure \ref{fig:intro-example-generic-fj}.
Our algorithm deducts those types out of the input in \ref{fig:intro-example-generic-jtx}.

But this is only possible by processing each class by itself as explained in chapter \ref{chapter:syntax}.
Lets assume for example, inside the class \texttt{Pair} there would be a call to the method \texttt{setfst} like so:
\texttt{setfst(new Integer())}.
This would cause the method \texttt{setfst} to get the type \texttt{Pair<Integer, Y> setfst(Integer fst)}.
This is due to the type rule \texttt{GT-CLASS} in figure \ref{fig:expression-typing}.
Inside the same class methods cannot be used in a polymorphic way.
We have to make this restriction to combat polymorphic recursion, which would render type inference undecidable.

While the overhead of explicit types look reasonable in the example,
realistic programs often contain variable initializations like
the following:\footnote{Taken from
  \url{https://stackoverflow.com/questions/4120216/map-of-maps-how-to-keep-the-inner-maps-as-maps/4120268}.} 
\begin{lstlisting}[basicstyle=\ttfamily\fontsize{8}{9.6}\selectfont,style=fgj]
  HashMap<String, HashMap<String, Object>> outerMap =
    new HashMap<String, HashMap<String, Object>>();
\end{lstlisting}

Java's \emph{local variable type inference} (since version 10\footnote{\url{https://openjdk.java.net/jeps/286}}) deals
satisfactorily with examples like the initialization of
\lstinline{outerMap}. 
In many initialization scenarios for local variables, Java infers their type
if it is obvious from the context. In the
example, we can write
\begin{lstlisting}[basicstyle=\ttfamily\fontsize{8}{9.6}\selectfont,style=fgj]
var outerMap = new HashMap<String, HashMap<String, Object>>();
\end{lstlisting}
because the constructor of the map spells out the type in
full. More specifically, ``obvious'' means that the right side of the initialization is
\begin{itemize}
\item a constant of known type (e.g., a string),
\item a constructor call, or
\item a method call (the return type is known from the method
  signature).
\end{itemize}
The \lstinline{var} declaration can also be used for an iteration
variable where the type can be obtained from the elements of the
container or from the initializer.
Alternatively, if the variable is used as the method's return value,
its type can be obtained from the current method's signature.

However, there are still many places where the programmer must provide types. In
particular, an explicit type must be given for
\begin{itemize}
\item a field of a class,
\item a local variable without initializer or initialized to \lstinline{NULL},
\item a method parameter, or
\item a method return type.
\end{itemize}

In this paper, we study \emph{global type inference} for Java. Our aim
is to write code that omits most type annotations, except for class
headers and field types. Returning to the \lstinline{Pair} example, it
is sufficient to write the code in Figure~\ref{fig:intro-example-generic-jtx}
and global type inference fills in the rest so that the result is
equivalent to Figure~\ref{fig:intro-example-generic-fj}. Our
motivation to study global type inference is threefold.
\begin{itemize}
\item Programmers are relieved from writing down obvious types. 
\item Programmers may write types that leak implementation details. The
  \lstinline{outerMap} example provides a good example of this
  problem. From a software engineering
  perspective, it would be better to use a more general abstract type like
\begin{lstlisting}[basicstyle=\ttfamily\fontsize{8}{9.6}\selectfont,style=fgj]
Map<String, Map<String, Object>> outerMap = ...
\end{lstlisting}
  Global type inference finds most general types.
\item Programmers may write types that are more specific than
  necessary instead of using generic types. Here, type
  inference helps programmers to find the most general type. Suppose
  we wanted to add a static  method \texttt{eqPair} for pairs of integers to the
  \lstinline/Pair/ class.
\begin{lstlisting}[basicstyle=\ttfamily\fontsize{8}{9.6}\selectfont,style=fgj]
boolean eqPair (Pair<Integer,Integer> p) {
  return p.fst.equals(p.snd);
}
\end{lstlisting}
  With global type inference it is sufficient to write the code on the
  left of Figure~\ref{fig:equal-pair} and obtain the FGJ code with the most general type on the right.
\end{itemize}
  \begin{figure}[t]
    \begin{minipage}[t]{0.49\linewidth}
\begin{lstlisting}[style=tfgj]
eqPair (p) {
  return p.fst.equals(p.snd);
}
\end{lstlisting}
    \end{minipage}
    \begin{minipage}[t]{0.49\linewidth}
\begin{lstlisting}[style=fgj]
<T> boolean eqPair (Pair<T,T> p){
  return p.fst.equals<T>(p.snd);
}
\end{lstlisting}
    \end{minipage}
    \caption{\lstinline{eqPair} in \TFGJ and FGJ}
    \label{fig:equal-pair}
  \end{figure}

To make our investigation palatable, we focus on global type inference for Featherweight
Generic Java \cite{DBLP:journals/toplas/IgarashiPW01} (FGJ), a
functional Java core language with full support for generics. Our type inference algorithm
applies to FGJ programs that specify the full class header and all field types,
but omit all method signatures. 
Given this input, our algorithm
infers a set of most general method signatures (parameter types and return types).
Inferred types are generic as much as possible and may contain
recursive upper bounds.

The inferred signatures have the following round-trip property
(relative completeness). If we
start with an FGJ program that does not make use of polymorphic
recursion (see Section~\ref{sec:polym-recurs}), strip all types from
method signatures, and run the algorithm on the 
resulting stripped program, then at least one of the inferred typings is equivalent or more
general than the types in the original FGJ program.

\subsection*{Contributions}
\label{sec:contributions}

We specify syntax and type system of the language \FGJGT, which drops all method type
annotations from FGJ and the typing of which rules out polymorphic
recursion. This language is amenable to polymorphic type inference and
each \FGJGT program can be completed to an FGJ program (see Section~\ref{sec:preliminaries}). 

We characterize uses of polymorphic recursion in FGJ and their impact
on signatures of generic methods (Section~\ref{sec:polym-recurs-form}). 

We define a constraint-based algorithm that performs global type
inference for \FGJGT. This algorithm is sound and relatively complete
for FGJ programs without polymorphic recursion
(Sections~\ref{sec:type-infer-algor} and~\ref{sec:unify}). Our algorithm improves
on previous attempts at type inference for Java in the literature as detailed in
Section~\ref{sec:related-work}. 

We investigate the complexity of global type inference and show its NP-completeness (Section~\ref{sec:complexity}).

We implemented a prototype of the type inference algorithm, which we
plan to submit for artifact evaluation.

\section{Motivation}
\label{sec:motivation}


This section presents a sequence of more and more challenging
examples for global type inference (GTI). To spice up our examples
somewhat, we assume some predefined utility classes with the following
interfaces.
\begin{lstlisting}[style=fgj]
class Bool {
  Bool not(); 
}
class Int {
  Int negate ();
  Int add (Int that);
  Int mult (Int that);
}
class Double {
  Double negate ();
  Double add (Double that);
  Double mult (Double that);
}
\end{lstlisting}

We generally use upper case single-letter identifiers like $\TVX,
\TVY, \dots$ for type variables.
Given a FGJ-GT class \CL 0, 
we call any FGJ class \CL i that can be transformed to \CL 0 by
erasing type annotations a \emph{completion of  \CL 0}.

\subsection{Multiplication}
\label{sec:multiplication}

Here is the \TFGJ code  for multiplying the components of a
pair.\footnote{We indicate \TFGJ code fragments by using a
  {gray background}.}
\begin{lstlisting}[style=tfgj]
class MultPair {
  mult (p) { return p.fst.mult(p.snd); }
}
\end{lstlisting}
Assuming the parameter typing $\mv p:\mv P$, result type $\mv R$, and that
\texttt{mult} in the body refers to \texttt{Int.mult}, we
obtain the following constraints.
\begin{itemize}
\item From \texttt{p.fst}: $\mv P \subconstr \mathtt{Pair}\Angle{\TVX,\TVY}$ and
  $\texttt{p.fst} : \TVX$.
\item From \texttt{p.snd}: $\mv P \subconstr \mathtt{Pair}\Angle{\TVZ,\TVW}$ and
  $\texttt{p.snd} : \TVW$.
\item The two constraints on $\mv P$ imply that $\TVX \eqconstr \TVZ$ and
  $\TVY \eqconstr \TVW$.
\item From \texttt{.mult (p.snd)}: $\TVX \subconstr \mathtt{Int}$, $\TVY \subconstr
  \mathtt{Int}$, and $\mathtt{Int} \subconstr \mv R$.
\end{itemize}
The return type $\mv R$ only occurs positively in the constraints, so we can
set $\mv R = \mathtt{Int}$.
The argument type $\mv P$ only occurs negatively in the constraints,
so $\mv P = \mathtt{Pair} \Angle{\TVX,\TVY}$.
This reasoning gives rise to the following completion.
\begin{lstlisting}[style=fgj]
class MultPair {
  <X extends Int, Y extends Int>
  Int mult (Pair<X,Y> p) { return p.fst.mult(p.snd); }
}
\end{lstlisting}
We obtain a second completion if we assume that \texttt{mult} refers to
\texttt{Double.mult}.
\begin{lstlisting}[style=fgj]
class MultPair {
  <X extends Double, Y extends Double>
  Double mult (Pair<X,Y> p) { return p.fst.mult(p.snd); }
}
\end{lstlisting}

Finally, the definition of \texttt{mult} might be recursive, which
generates different constraints for the method invocation of \texttt{mult}.
\begin{itemize}
\item From \texttt{.mult (p.snd)}: $\TVX \subconstr \mathtt{MultPair}$, $\TVY \subconstr
  \mathtt{P}$, and $\mathtt{R} \subconstr \mv R$.
\end{itemize}
Transitivity of subtyping applied to $\TVY \subconstr \mathtt{P}$ and $\mv P \subconstr \mathtt{Pair}\Angle{\TVX,\TVY}$
yields the constraint $\TVY \subconstr \mathtt{Pair}\Angle{\TVX,\TVY}$, which triggers the
occurs-check in unification and is hence rejected. 


The two solutions can be combined to
\begin{lstlisting}[style=fgj]
class MultPair {
  <X extends T1, Y extends T2>
  T0 mult (Pair<X,Y> p) { return p.fst.mult(p.snd); }
}
\end{lstlisting}
where
\begin{align*}
  (T_0, T_1, T_2) & \in \{ (\mv{Int}, \mv{Int}, \mv{Int}), (\mv{Double}, \mv{Double}, \mv{Double}) \}
\end{align*}

\subsection{Inheritance}
\label{sec:inheritance}

\begin{figure}[tp]
  \begin{subfigure}[t]{0.49\linewidth}
\begin{lstlisting}[style=tfgj]
class A1 {
  m(x) { return x.add(x); }
}
class B1 extends A1 {
  m(x) { return x; }
}
\end{lstlisting}
  \end{subfigure}
\begin{subfigure}[t]{0.49\linewidth}
\begin{lstlisting}[style=tfgj]
class A2 {
  m(x) { return x; }
}
class B2 extends A2 {
  m(x) { return x.add(x); }
}
\end{lstlisting}
  \end{subfigure}
  \caption{Method overriding}
  \label{fig:method-overriding}
\end{figure}

Let's start with the artificial example in the left
listing of Figure~\ref{fig:method-overriding} and ignore the \texttt{Double} class. Type
inference proceeds according 
to the inheritance hierarchy starting from the superclasses. In class
\texttt{A1}, the inferred method type is \texttt{Int A1.m (Int)}. Class \texttt{B1} is a
subclass of \texttt{A1} which must override \texttt{m} as there is no
overloading in FGJ. However, the inferred method 
type is \texttt{<T> T B1.m(T)}, which is not a correct
method override for \texttt{A1.m()}.
Hence, GTI must instantiate the method type in the subclass \texttt{B1} to
\texttt{Int B1.m(Int)}.

Conversely, for the right listing of
Figure~\ref{fig:method-overriding}, GTI infers the types
\texttt{<T> T A2.m (T)} and \texttt{Int B2.m (Int)}. Again, these
types do not give rise to a correct method override and 
GTI is now forced to instantiate the type in the superclass to
\texttt{Int A2.m (Int)}.

\todo[inline]{This example shows that GTI must first collect the constraints 
  for all methods \texttt{m{}} in a class hierarchy. Generalization can only happen after
  all constraints on \texttt{m}'s type have been considered.}


In full Java, type inference would have to offer two alternative
results: either two different
overloaded methods (one inherited and one local) in \texttt{B1}/\texttt{B2} or
impose the typing \texttt{Int B1.m(Int)} or \texttt{Int A2.m(Int)} to enforce correct overriding.

\subsection{Inheritance and Generics}
\label{sec:inheritance-generics}

\begin{lstlisting}[float,caption={Function class}, label={lst:function-class},style=fgj]
class Function<S,T> {
  T apply(S arg) { return this.apply (arg); }
}
\end{lstlisting}
Suppose we are given a generic class for modeling functions in  FGJ (Listing~\ref{lst:function-class}).
This code is constructed to serve as an ``abstract'' super class to derive more
interesting subclasses.
The class \texttt{Function<S,T>} must be presented in this explicit
way. Its type annotations \textbf{cannot} be inferred by GTI because
the use of the generic class parameters in the method type cannot be inferred from the
implementation.

If we applied GTI to the type-erased version of
Listing~\ref{lst:function-class}, the \texttt{apply} method would be considered a generic method: 
\begin{center}
  \begin{minipage}{0.3\linewidth}
\begin{lstlisting}[style=tfgj]
apply (arg) { ... }
\end{lstlisting}
  \end{minipage}
  \hfill\texttt{ --GTI--> }\hfill
  \begin{minipage}{0.45\linewidth}
\begin{lstlisting}[style=fgj]
<A,B> B apply (A arg) { ... }
\end{lstlisting}
  \end{minipage}
\end{center}
The typing of \texttt{apply} in Listing~\ref{lst:function-class} is an
instance of this result, so that completeness of GTI is preserved!


Now that we have the abstract class \texttt{Function<S,T>} at our
disposal, let us apply GTI to a class of boxed values with a
\texttt{map} function:
\begin{lstlisting}[style=tfgj]
class Box<S> {
  S val;
  map(f) {
    return new Box<>(f.apply(this.val));
} }
\end{lstlisting}
GTI finds the following constraints
\begin{itemize}
\item the return value must be of type \texttt{Box<T>}, for some type
  \texttt{T},
\item \texttt{T} is a supertype of the type returned by
  \texttt{f.apply},
\item \texttt{apply} is defined in class \texttt{Function<S1,T1>} with
  type \texttt{T1 apply(S1 arg)}, 
\item hence \texttt{T1\,<:\,T} and \texttt{S\,<:\,S1} (because
  \texttt{this.val\,:\,S}),
\end{itemize}
and resolves them to the desired outcome where \texttt{T1=T} and
\texttt{S=S1} using the methods of
Simonet~\cite{DBLP:conf/aplas/Simonet03}. 
\begin{lstlisting}[style=fgj]
class Box<S> {
  S val;
  <T> Box<T> map(Function<S,T> f) {
    return new Box<T>(f.apply<S,T>(this.val));
} }
\end{lstlisting}
But what happens if we add subclasses of \texttt{Function}?
For example:
\begin{lstlisting}[style=tfgj]
class Not extends Function<Bool,Bool> {
  apply(b) { return b.not(); }
}
class Negate extends Function<Int,Int> {
  apply(x) { return x.negate(); }
}
\end{lstlisting}
\todo[inline]{Does this approach really make sense? If we have a
  subclass that overrides a method of the superclass, then the public
  interface should be the one of the superclass. Subsequently, the
  implementation of the method in the subclass should be checked
  against the type inferred for the superclass.}
If we rerun GTI with these classes, we now have additional
possibilities to invoke the \texttt{apply} method. With \texttt{Not}, we need to use the
generic type of \texttt{Function.apply()}, but instantiate it according to
\texttt{Function<Bool,Bool>}. Thus,
we obtain the constraints \texttt{Bool $\subconstr$ T} and \texttt{S $\subconstr$ Bool} for \texttt{T = Bool}
and \texttt{S = Bool}, which are both satisfiable. With
\texttt{Negate} we run into the same situation with the constraints
\texttt{Int $\subconstr$ Int} and \texttt{Int $\subconstr$ Int}.

Here is another subclass of \texttt{Function<S,T>} that we want
to consider.
\begin{lstlisting}[style=fgj]
class Identity<S> extends Function<S,S> {
  S apply(S arg) { return arg; }
}
\end{lstlisting}
Here, we obtain the following type constraints
\begin{itemize}
\item \texttt{apply} is defined in class \texttt{Identity<S1>} with
  type \texttt{S1 apply (S1 arg)},
\item hence \texttt{S1 $\subconstr$ T} and \texttt{S $\subconstr$ S1}.
\end{itemize}
Resolving the constraints yields \texttt{S = T} thus the typing
\begin{lstlisting}[style=fgj]
Box<S> map(Identity<S> f);
\end{lstlisting}
which is an instance of the previous typing.

\subsection{Multiple typings}
\label{sec:multiple-results}
\begin{figure}[tp]
  \begin{minipage}{0.49\linewidth}
\begin{lstlisting}[style=fgj]
class List<A> {
  List<A> add(A item) {...}
  A get() { ... }
}
\end{lstlisting}
  \end{minipage}
  ~$\left|
  \begin{minipage}{0.49\linewidth}
\begin{lstlisting}[style=tfgj]
class Global{
  m(a){
    return a.add(this).get();
} }
\end{lstlisting}
  \end{minipage}\right.$
  \caption{Example for multiple inferred types}
  \label{fig:example-types-not-unique}
\end{figure}
Global type inference processes classes in order of
dependency.
To see why, consider the classes \texttt{List<A>} and \texttt{Global}
in Figure~\ref{fig:example-types-not-unique}. 
Class \texttt{Global} may depend on
class \texttt{List} because \texttt{Global} uses methods \texttt{add} and
\texttt{get} and \texttt{List} defines methods with the same names.
The dependency is only approximate because, in general, there may be additional classes
providing methods \texttt{add} and \texttt{get}.

In the example, it is safe to assume that the types for the methods of class \texttt{List}
are already available, either because they are given (as in the code
fragment) or because they were inferred before considering class \texttt{Global}.

The method \texttt{m} in class \texttt{Global} first invokes
\texttt{add} on \texttt{a}, so the type of \texttt{a} as well as the
return type of \texttt{a.add(this)} must be
\texttt{List<T>}, for some \texttt{T}. As \texttt{this} has
type \texttt{Global}, it must be that \texttt{Global} is a
subtype of \texttt{T}, which gives rise to the constraint
\texttt{Global $\subconstr$ T}. By the typing of \texttt{get()} we
find that the return type of method \texttt{m} is also \texttt{T}.

But now we are in a dilemma because FGJ only supports \emph{upper bounds} for
type variables,\footnote{Java has the same restriction. Lower bounds
  are only allowed for wildcards.} so that 
\texttt{Global $\subconstr$ T} is not a valid constraint in FGJ.
To stay compatible with this restriction, global type inference
expands the constraint by instantiating \texttt{T} with the (two) superclasses fulfilling
the constraint, \texttt{Global} and \texttt{Object}.  They give rise to two incomparable
types for 
\texttt{m}, \texttt{List<Global> -> Global} and
\texttt{List<Object> -> Object}. So there are two different FGJ
programs that are completions of the \texttt{Global} class.

GTI models these instances by inferring an \emph{intersection type}
\texttt{List<Global> -> Global \& List<Object> -> Object}
for method \texttt{m} and the different FGJ-completions of class \texttt{Global} are
instances of the intersection type:\footnote{The cognoscenti will be
  reminded of overloading. However, FGJ does not support overloading,
  so we rely on resolution by subsequent uses of the method. Moreover,
  this intersection type cannot be realized by overloading in a Java
  source program because it is resolved according to the raw classes
  of the arguments, in this case \lstinline{List}. It can be realized
  in bytecode which supports overloading on the return type, too.%
}
\begin{center}
  \begin{minipage}{0.49\linewidth}
\begin{lstlisting}[style=fgj]
class Global {
  Global m(List<Global> a) {
    return a.add(this).get();
}
\end{lstlisting}
  \end{minipage}
  \begin{minipage}{0.49\linewidth}
\begin{lstlisting}[style=fgj]
class Global {
  Object m(List<Object> a) {
    return a.add(this).get();
}
\end{lstlisting}
  \end{minipage}
\end{center}
In this sense, the inferred intersection type represents a principal
typing for the class.
Additional classes in the program may further restrict the number of
viable types. Suppose we define a class \texttt{UseGlobal} as
follows:
\begin{lstlisting}[style=tfgj]
class UseGlobal {
  main() {
    return new Global().m((List<Object>) new List());
} }
\end{lstlisting}
Due to the dependency on \texttt{Global.m()}, type inference considers this class after class
\texttt{Global}. As it uses \texttt{m} at type
\texttt{List<Object> -> Object}, global type inference narrows the
type of \texttt{m} to just this alternative.


\subsection{Polymorphic recursion}
\label{sec:polym-recurs}
\begin{figure}[tp]
  \begin{minipage}{0.49\linewidth}
\begin{lstlisting}[style=fgj]
class UsePair {
  <X,Y> Object prc(Pair<X,Y> p) {
    return this.prc<Y,X> (p.swap<X,Y>());
} }
\end{lstlisting}
  \end{minipage}
  ~$\left|
  \begin{minipage}{0.49\linewidth}
\begin{lstlisting}[style=tfgj]
class UsePair {
  prc(p) {
    return this.prc (p.swap());

} }
\end{lstlisting}
  \end{minipage}\right.$
  \caption{Example for polymorphic recursion}
  \label{fig:examples-poly-rec}
\end{figure}
A program uses \emph{polymorphic recursion} if there is a generic method that is invoked
recursively at a more specific type than its definition.
As a toy example for polymorphic recursion consider the FGJ class \texttt{UsePair} with a
generic method \texttt{prc} that invokes itself
recursively on a swapped version of its argument pair
(Figure~\ref{fig:examples-poly-rec}, left).
This method makes use of polymorphic recursion because the type of the
recursive call is different from the declared type of the method. More
precisely, the declared argument type is \texttt{Pair<X,Y>} whereas
the argument of the recursive call has type
\texttt{Pair<Y,X>}---an instance of the declared type.

For this particular example, global type inference succeeds on the
corresponding stripped program shown in
Figure~\ref{fig:examples-poly-rec}, right, but it yields a more restrictive
typing of \texttt{<X> Object prc (Pair<X,X> p)} for the method. 
A minor variation of the FGJ program with a non-variable instantiation makes type inference fail entirely:
\begin{lstlisting}[style=fgj]
class UsePair2 {
  <X,Y> Object prc(Pair<X,Y> p) {
    return this.prc<Y,Pair<X,Y>> (new Pair (p.snd, p));
  }
}
\end{lstlisting}

Polymorphic recursion is known to make type inference intractable
\cite{DBLP:journals/toplas/Henglein93,DBLP:journals/toplas/KfouryTU93}
because it can be reduced to an undecidable semi-unification problem
\cite{DBLP:journals/iandc/KfouryTU93}. However, \emph{type checking} with
polymorphic recursion is tractable and routinely used in languages like Haskell
and Java.

GTI does not infer method types with polymorphic recursion. Inference either fails or
returns a more restrictive type. Classes making use of polymorphic recursion need to supply
explicit typings for methods in question.



\section{Featherweight Generic Java with Global Type Inference}
\label{sec:preliminaries}

This section defines the syntax and type system of a modified version
of the language Featherweight Generic Java
(FGJ)~\cite{DBLP:journals/toplas/IgarashiPW01}, which we call \TFGJ
(with Global Type Inference). The main omissions with respect to FGJ are method types specifications
and polymorphic recursion. We finish the section by formally
connecting FGJ and \TFGJ and by establishing some properties about
polymorphic recursion in FGJ.

\subsection{Syntax}\label{chapter:syntax}
\begin{figure}[tp]
\begin{align*}
  \mv T &::= \mv X \mid \mv N \\
  \mv N &::= \exptype{C}{\ol{T}}\\
  \mv L &::= \mathtt{class} \ \exptype{C}{\ol{X} \triangleleft \ol{N}} \triangleleft \ \mv N\ \{ \ol{T} \ \ol{f}; \,\mv K \, \ol{M} \} \\
  \mv K &::= \mv C(\ol{f})\ \{\mathtt{super}(\ol{f}); \ \mathtt{this}.\ol{f}=\ol{f};\} \\
  \mv M &::= \mathtt{m}(\ol{x})\ \{ \mathtt{ return}\ \mv e; \} \\
  \mv e &::= \mv x \mid \mv e.\mv f \mid
             \mv e.\mathtt{m}(\ol{e}) \mid \mathtt{new}\ \mathtt{C}(\ol{e})
             \mid (\mv N)\ \mv e
\end{align*}
  \caption{Syntax of \TFGJ}
  \label{fig:syntax-tfgj}
\end{figure}
Figure~\ref{fig:syntax-tfgj} defines the syntax of \TFGJ.
Compared to FGJ,
type annotations for method parameters and method return types are omitted.
Object creation via \texttt{new} as well as method calls come do not
require instantiation of their generic parameters.
We keep the class constraints ${\ol{X} \triangleleft \ol{N}}$ as well as the types
for fields $\ol{T} \ \ol{f}$ as we consider them as part of the
specification of a class.

We make the following assumptions for the input program:
\begin{itemize}
\item All types $\mv N$ and $\mv T$ are well formed according to the
  rules of FGJ, which carry over to \TFGJ (see Fig.~\ref{fig:well-formedness-and-subtyping}).
\item The methods of a class call each other mutually recursively.
\item The classes in the input are topologically sorted so that later
  classes only call methods in classes that come earlier in the
  sorting order.
\end{itemize}
Our requirements on the method calls do not impose serious
restrictions as any class, say \mv{C}, can be transformed to meet them as
follows. A preliminary dependency analysis determines an 
approximate call graph. We cluster the methods of \mv{C} according to
the $n$ strongly
connected components of the call graph. Then we split the class into a
class hierarchy 
$\mv{C}_1 \extends \dots \extends \mv{C}_n$ such that each class $\mv{C}_i$ contains
exactly the methods of one strongly connected component and assign a
method cluster to $\mv{C}_i$ if all calls to methods of  $\mv{C}$ now
target methods assigned to $\mv{C}_j$, for some $j\ge i$. The class
$\mv{C}_1$ replaces $\mv{C}$ everywhere in the program: in subtype
bounds, in \texttt{new} expressions, and in casts. More precisely, if
\mv{C} is defined by $\mathtt{class\ \exptype{C}{\ol X \extends \ol N}
\extends N \dots} $, then the class headers for the $\mv{C}_i$ are
defined as follows:
\begin{itemize}
\item  $\mathtt{class\ \exptype{C_i}{\ol X \extends \ol N}
\extends \exptype{C_{i+1}}{\ol X} \dots} $, for $1\le i < n$ and
\item  $\mathtt{class\ \exptype{C_n}{\ol X \extends \ol N}
\extends N \dots} $.
\end{itemize}
It follows from this discussion that the resulting classes have to be
processed backwards starting with $\mv{C}_n, \mv{C}_{n-1}, \dots, \mv{C}_1$.
Figure \ref{fig:example-decluster} showcases this process with a short example.

\begin{figure}[tp]
    \begin{subfigure}[t]{0.49\linewidth}
\begin{lstlisting}[style=tfgj]
class C extends Object {
  m1(a){
    return a;
  }
  m2(b){
    return this.id(a);
  }
}
\end{lstlisting}
      \caption{The methods \texttt{m1} and \texttt{m2} can be separated}
    \end{subfigure}
    ~
    \begin{subfigure}[t]{0.49\linewidth}
\begin{lstlisting}[style=tfgj]
class C1 extends C2 {
  m2(b){
    return this.id(a);
  }
}
class C2 extends Object {
  m1(a){
    return a;
  }
}
\end{lstlisting}
      \caption{After the transformation}
    \end{subfigure}
    \caption{Example for splitting a class into its strongly connected components}
    \label{fig:example-decluster}
  \end{figure}

\subsection{Typing}
\label{chapter:type-rules}
We start with some notation.
An environment $\mv\Gamma$ is a finite mapping from variables to types,
written $\ol x:\ol T$; a type environment $\mv\Delta$ is a finite mapping
from type variables to nonvariable types, written $\ol X\subeq\ol
N$, which takes each type variable to its bound. As in FGJ, we do not
impose an ordering on environment entries to enable F-bounded
polymorphism.

There is a new method environment $\mv\Pi$ which maps pairs of a class header
$\exptype{C}{\ol X}$ and a method name $\mv m$ to a set of method
types of the form $\exptype{}{\ol{Y} \triangleleft  \ol{P}} \ol{T} \to \mv{T}$. 
It supports the \textit{mtype} function that relates a nonvariable type
$\mv N$ and a method name $\mv m$ to a method type.

The judgments for subtyping $\mathtt{\Delta \vdash S \subeq T}$ and
well-formedness of types $\mathtt{\Delta \vdash T\ \mathtt{ok}}$
(Figure~\ref{fig:well-formedness-and-subtyping}) stay the same as in
FGJ.

The overall approach to typing changes with respect to FGJ. In FGJ,
classes can be checked in any order as the method typings of all other
classes are available in the syntax. \TFGJ processes classes in
order such that early classes do not invoke methods in late classes.

The new program typing rule \rulename{GT-PROGRAM} for the judgment
$\vdash \ol{L} : \mv\Pi$ reflects this
approach. It starts with an empty
method environment and applies class typing to each class in the
sequence provided. Each processed class adds its method typings to the method
environment which is threaded through to constitute the program type
as the final method environment $\mv\Pi$.

Expression typing $\mathtt{\Pi; \Delta; \Gamma \vdash e
  : T}$ changes subtly (see Figure~\ref{fig:expression-typing}). As
\TFGJ omits some type annotations, we are forced to adapt some of  
FGJ's typing rules. The new rules infer omitted types and disable
polymorphic recursion. 

The new method 
environment $\mv\Pi$ is only used in the revised
rule for method invocation 
\rulename{GT-INVK}, where it is passed as an additional parameter to
\textit{mtype}. The revised definition of \textit{mtype} (Figure~\ref{fig:auxiliary-functions}) locates the class that contains
the method definition by traversing the subtype hierarchy and looks up the
method type in environment $\mv\Pi$, which contains the method types
that were already inferred. Our definition of \textit{mtype} does not support
overloading as $\mv\Pi$ relate at most one type to each method
definition (cf.\ rule \rulename{GR-CLASS}). The instantiation of the
method's type parameters is inferred in \TFGJ.

The rule \rulename{GT-NEW} changes to infer the
instantiation of the class's type parameters: the rule
simply assumes a suitable instantiation by some $\ol U$.

Finally, \rulename{GT-CAST} replaces the three rules
\rulename{GT-UCAST'}, \rulename{GT-DCAST'}, and \rulename{GT-SCAST'}
of FGJ. This is a slight simplification with respect to FGJ. While
the three original rules cover disjoint use cases (upcast, downcast,
and stupid cast that is sure to fail) of the cast
operation, they are not exhaustive! The rule \rulename{GT-DCAST'} only
admits downcasts that work the same in a type-passing semantics as in
a type erasure semantics. We elide this distinction for simplicity,
though it could be handled by introducing constraints analogous to the
\textit{dcast} function from FGJ. 



The typing rule for a method $\mv m$, \rulename{GT-METHOD}, changes
significantly. By our assumption on the order, in which classes are
processed, the typing of $\mv m$ is already provided by the 
method environment $\mv\Pi$.  The type environment $\mv\Delta$ is
also provided as an input. Moreover, to rule out polymorphic
recursion, the assumptions about the local methods of class $\mv{C}$
are monomorphic at this stage. The rule type checks the body for the
inferred type of method $\mv m$.

All this information is provided and generated by the rule for class
typing, \rulename{GT-CLASS}. A class typing for \mv{C} receives an incoming
method type environment $\mv\Pi$ and generates an extended one
$\mv{\Pi''}$ which additionally contains the method types inferred for
\mv{C}.

In $\mv{\Pi'}$, we generate some monomorphic types for all
methods of class $\mv C$. We use these types to check the methods. Afterwards, we
return generalized versions of these same types in $\mv{\Pi''}$. All
method types use the same generic type variables $\ol Y$ with the
same constraints $\ol P$. It is safe to make this assumption in the absence of
polymorphic recursion as we will show in
Proposition~\ref{prop:polymorphi-recursion}.

\begin{figure}[tp]
  \fbox{
    \begin{minipage}{0.9\textwidth }
      \begin{small}
        \textbf{Subtyping:}\\[1em]
        \begin{tabularx}{\linewidth}{X c X r}
          & $\mathtt{ \Delta \vdash T \subeq  T } $
          &   & \rulename{S-REFL} \\

          & \\
          &
          $\mathtt{\ddfrac{ \Delta \vdash S \subeq  T \quad \quad
              \Delta \vdash T \subeq  U }{ \Delta \vdash S \subeq  U
            }}$ & & \rulename{S-TRANS} \\

          & \\

          & $\mathtt{ \Delta \vdash X \subeq  \Delta(X)
          }$ & & \rulename{S-VAR} \\
          & \\
          &
          $\mathtt{\ddfrac{ \texttt{class}\ \exptype{C}{\ol{X}
                \triangleleft \ol{N}} \triangleleft \mv N \set{ \ldots
              } }{ \Delta \vdash \exptype{C}{\ol{T}} \subeq 
              [\ol{T}/\ol{X}]\mv N
            }}$ & & \rulename{S-CLASS} \\
          & & & \\
          \hline
          \multicolumn{1}{l}{\textbf{Well-formed types:}}\\
          & & & \\
          & $\mathtt{ \Delta \vdash \texttt{Object}\ \texttt{ok}
          }$ & & \rulename{WF-OBJECT}\\
          & \\
          &
          $\mathtt{\ddfrac{ X \in \textit{dom}(\Delta) }{ \Delta
              \vdash X \ \texttt{ok} } }
          $ & & \rulename{WF-VAR} \\
          & \\
          & $\mathtt{\ddfrac{\begin{array}{c}
                               \texttt{class}\ \exptype{C}{\ol{X}
                               \triangleleft \ol{N}} \triangleleft \mv N \{ \ldots \} \\
                               \mathtt{\Delta} \vdash \ol{T} \ \texttt{ok}
                               \quad \quad \mathtt{\Delta} \vdash \ol{T} \subeq 
                               [\ol{T}/\ol{X}]\ol{N}
                             \end{array}
                           }{ \Delta \vdash \exptype{C}{\ol{T}} \
                             \texttt{ok} } } $ & & \rulename{WF-CLASS}
                       \end{tabularx}
                     \end{small}
                   \end{minipage}
                 }
                 \caption{Well-formedness and subtyping}
                 \label{fig:well-formedness-and-subtyping}
               \end{figure}

\begin{figure}[tp]
\fbox{
\begin{minipage}{0.9\textwidth}
\begin{small}
\textbf{Expression typing:}\\
\begin{tabularx}{\textwidth}{c X r}
  $\mathtt{
\environmentvdash x : \Gamma(x)
}$ & & \rulename{GT-VAR} \\
& \\

$\mathtt{\ddfrac{
    \Pi; \Delta; \Gamma \vdash e_0:T_0 \qquad
    \mathit{fields}(\mathit{bound}_\Delta(T_0)) = \overline{T} \ \overline{f}}
  {\Pi; \Delta; \Gamma \vdash e_0.\mathtt{f}_i : T_i}
}
$ & & \rulename{GT-FIELD} \\
& \\

$\mathtt{\ddfrac{\begin{array}{c}
  \mathtt{\environmentvdash e_0 : T_0 } \quad \quad 
  \mathtt{\mathit{mtype}(m, \mathit{bound}_\Delta (T_0), \Pi) = \exptype{}{\ol{Y} \triangleleft \ol{P}} \ol{U} \to U } \\
                   \mathtt{\Delta \vdash \ol{V} \ \texttt{ok} } \quad \quad
                   \mathtt{\Delta \vdash \ol{V} \subeq  [\ol{V}/\ol{Y}]\ol{P} } \qquad
  \mathtt{\environmentvdash \ol{e} : \ol{S} } \quad \quad
  \mathtt{\Delta \vdash \ol{S} \subeq  [\ol{V}/\ol{Y}]\ol{U}}
\end{array}}
{\environmentvdash \mathtt{e_0.\mv{m}(\overline{e}) : [\ol{V}/\ol{Y}]U }}
}$ & & \rulename{GT-INVK} \\
& \\
$\mathtt{ \ddfrac{\Delta \vdash {\mv N} \ \texttt{ok} \quad 
  \mv N = \exptype{C}{\ol{U}} \quad
  \textit{fields}(\mv N) = \ol{T}\ \ol{f} \quad 
  \environmentvdash \ol{e} : \ol{S} \quad \Delta \vdash \ol{S} \subeq  \ol{T}
}{
  \environmentvdash \texttt{new C}(\ol{e}): \mv N
}
}$ & & \rulename{GT-NEW} \\

& \\

$\ddfrac{\mathtt{\environmentvdash e_0 : T_0}}
{\mathtt{\environmentvdash (N) e_0 : N}}
$ & & \rulename{GT-CAST} \\
%
%
%
%
%
%
\end{tabularx}\\[1em]
\textbf{Method typing:}\\[1em]
\begin{tabularx}{\textwidth}{c X r}
  $\ddfrac{
    \begin{array}{c}
      \mathtt{\forall  \ol{T}, T: \exptype{}{} \ol{T} \to T \in
      \Pi(\exptype{C}{\ol{X} \triangleleft \ol{N}}.m) } \quad \quad
      \mathtt{\Delta \vdash S \subeq  T }
      \\
      \mathtt{\Pi; \Delta ; \ol{x}:\ol{T},\ this : \exptype{C}{\ol{X}} \vdash e_0 : S} 
      \\
      \mathtt{
      \textit{override}(m, N, \exptype{}{\ol{Y} \triangleleft
      \ol{P}}\ol{T} \to T, \Pi) } 
    \end{array}} {
    \mathtt{\Pi, \Delta \vdash  \texttt{m}(\ol{\mathtt{x}}) \{\texttt{return}\ \mathtt{e}_0;\}
      \texttt{ OK in }\exptype{C}{\ol{X} \triangleleft \ol{N}} \extends N \texttt{
        with } \exptype{}{\ol{Y} \triangleleft  \ol{P}} }
  }$ & & \rulename{GT-METHOD}\\
\end{tabularx}\\[1em]
\textbf{Class typing:}\\[1em]
  \begin{tabularx}{\textwidth}{X c X r}
  & $\ddfrac{
    \begin{array}{c}
      \mathtt{\Pi' = \Pi \cup \set{\exptype{C}{\ol{X} \triangleleft \ol{N}}.m \mapsto \exptype{}{} \ol{T_m} \to T_m \mid m \in \ol{M}} } \\
      \mathtt{\Pi'' = \Pi \cup \set{\exptype{C}{\ol{X}\triangleleft \ol{N}}.m \mapsto
      \exptype{}{\ol{Y} \triangleleft  \ol{P}} \ol{T_m} \to T_m \mid m \in \ol{M}} } \\
      \mathtt{\Delta = \ol{X} \subeq  \ol{N}, \ol{Y} \subeq  \ol{P}}
      \qquad \mathtt{\Delta \vdash\ol{P} \
      \texttt{ok}}
      \qquad \forall \mv{m}: \mathtt{\Delta \vdash \ol{T_m}, T_m \ \texttt{ok}}\\
      \mathtt{\ol{X} \subeq  \ol{N} \vdash \ol{N}, N, \ol{T}\ \texttt{ok}
      \quad\quad \mathit{fields}(\mathtt{N}) = \ol{\mathtt{U}} \ \ol{\mathtt{g}}} \\\
      \mathtt{\Pi', \Delta \vdash \ol{M} \ \texttt{OK IN}\
      \exptype{C}{\ol{X} \triangleleft \ol{N}} \extends N  \texttt{
    with } \exptype{}{\ol{Y} \triangleleft  \ol{P}}}\\
      \mathtt{K = C(\overline{U} \ \overline{g}, \overline{T} \ \overline{f}) \{ \texttt{super}(\overline{g}); \ \texttt{this}.\overline{f}=\overline{f}; \} }\\
  \end{array}
    }
  {\mathtt{ \Pi \vdash \texttt{class}\ \exptype{C}{\ol{X}
        \triangleleft \ol{N}} \triangleleft N\ \{ \overline{T} \
      \overline{f}; \ K \ \overline{M} \} \ \texttt{OK} : \Pi''}}
  $ & & \rulename{GT-CLASS} \\
\end{tabularx}\\[1em]
\textbf{Program typing:}\\[1em]
  \begin{tabularx}{\textwidth}{X c X r}
  & $\ddfrac{
    \mathtt{\emptyset \vdash L_1 : \Pi_1} \quad
    \mathtt{\Pi_1 \vdash L_2 : \Pi_2} \quad \dots \quad
    \mathtt{\Pi_{n-1} \vdash L_n : \Pi_n}
  }{
    \mathtt{\vdash \ol L : \Pi_n}
  }$
  && \rulename{GT-PROGRAM}
\end{tabularx}
\end{small}
\end{minipage}
}
\caption{Typing rules}
  \label{fig:expression-typing}
\end{figure}

\begin{figure}[tp]
\fbox{
\begin{minipage}{0.9\textwidth}
  \begin{small}
  \textbf{Field lookup:} \\[1em]
  \begin{tabularx}{\textwidth}{cXr}
    $\mathit{fields}(\mv{Object}) = \bullet$
    && \rulename{F-OBJECT} \\
    && \\
    $\ddfrac{
      \mathtt{\texttt{class}\ \exptype{C}{\ol{X} \triangleleft \ol{N}}\triangleleft
        \ol{N}\ \{ \overline{S} \ \overline{f}; \ K \ \overline{M} \}
      } \qquad
      \mathtt{\mathit{fields} ([\ol T/\ol X]N) = \ol U\ \ol g}
    }{
      \mathit{fields}(\exptype{C}{\ol T}) = \ol U\ \ol g, [\ol T/\ol
      X]\ol S\ \ol f
    }$
    && \rulename{F-CLASS}
  \end{tabularx}\\[1em]
  \textbf{Method type lookup:} \\[1em]
\begin{tabularx}{\textwidth}{cXr}
  $\ddfrac{
    \begin{array}{c}
      \mathtt{\texttt{class}\ \exptype{C}{\ol{X} \triangleleft
      \ol{N}}\ \triangleleft \mv{N}\ \{ \overline{C} \ \overline{f}; \ K \ \overline{M} \} 
      \qquad m \in \ol{M}} \\
      \mathtt{
      \exptype{}{\ol{Y} \triangleleft  \ol{P}} \ol{U} \to U \in \Pi (\exptype{C}{\ol{X} \triangleleft \ol{N}}.m) }
\end{array}
   } {\mathtt{
     \textit{mtype}(m, \exptype{C}{\ol{T}}, \Pi) = [\ol{T}/\ol{X}]\exptype{}{\ol{Y} \triangleleft \ol{P}}\ol{U} \to U
    }}$ & & \rulename{MT-CLASS}\\
 & & \\
  $\ddfrac{\mathtt{\texttt{class}\ \exptype{C}{\ol{X} \triangleleft \ol{N}}\triangleleft
  \mv{N}\ \{ \overline{C} \ \overline{f}; \ K \ \overline{M} \} 
  \qquad m \notin \ol{M}} }{
    \mathtt{\textit{mtype}(m, \exptype{C}{\ol{T}}, \Pi) =
      \textit{mtype}(m, [\ol{T}/\ol{X}]N, \Pi)}
    }$ & & \rulename{MT-SUPER} \\
\end{tabularx}\\[1em]
\textbf{Valid method overriding:}\\[1em]
\begin{tabularx}{\textwidth}{c}
  $\ddfrac{\mathtt{\textit{mtype}(m, N, \Pi) = \exptype{}{\ol{Z} \triangleleft \ol{Q}} \ol{U} \to \mv U \
    \text{implies}\ \ol{P}, \ol{T} = [\ol{Y}/\ol{Z}](\ol{Q},\ol{U}) \ 
   \text{and}\ \ol{Y} \subeq  \ol{P} \vdash T_0 \subeq  [\ol{Y}/\ol{Z}]U_0} 
  } {
    \mathtt{\textit{override}(m, N, \exptype{}{\ol{Y} \triangleleft \ol{P}} \ol{T} \to T_0, \Pi)}
  }$ 
\end{tabularx}
\end{small}
\end{minipage}
}
\caption{Auxiliary functions}
  \label{fig:auxiliary-functions}
\end{figure}

\subsection{Soundness of Typing}
\label{sec:soundness-typing}
\begin{figure}[tp]
    \begin{align*}
      \Erase{\mv x} &= \mv x \\
      \Erase{\mv e.\mv f} &= \Erase{\mv e}.\mv f \\
      \Erase{\exptype{e}{\ol T}.\mathtt{m}(\ol{e})} &= \Erase{\mv e}. \mathtt{m} (\Erase{\ol e}) \\
      \Erase{\mathtt{new}\ \exptype{C}{\ol T}(\ol{e})} & = \mathtt{new}\ \mv{C}(\Erase{\ol{e}}) \\
      \Erase{(\mv N)\ \mv e} & = (\mv N)\ \Erase{\mv e} \\
      \Erase{\exptype{}{\ol{X} \triangleleft \ol{N}}\ \mv{T}\ \mathtt{m}(\ol T\ \ol{x})\ \{ \mathtt{
      return}\ \mv e; \}} & = \mathtt{m}(\ol{x})\ \{ \mathtt{ return}\ \Erase{\mv e}; \} \\
      \Erase{\mv C(\ol{U}\ \ol{g}, \ol{T}\ \ol{f})\ \{\mathtt{super}(\ol{g}); \ \mathtt{this}.\ol{f}=\ol{f};\}} & = \mv C(\ol{g}, \ol{f})\ \{\mathtt{super}(\ol{g}); \ \mathtt{this}.\ol{f}=\ol{f};\} \\
      \Erase{\mathtt{class} \ \exptype{C}{\ol{X} \triangleleft \ol{N}} \triangleleft \ \mv N\ \{ \ol{T} \ \ol{f}; \,\mv K \, \ol{M} \}} & = 
                                                                                                                                          \mathtt{class} \ \exptype{C}{\ol{X} \triangleleft \ol{N}} \triangleleft \ \mv N\ \{ \ol{T} \ \ol{f}; \,\Erase{\mv K} \, \Erase{\ol{M}} \}
    \end{align*}
    \caption{Erasure functions}
    \label{fig:erasure}
  \end{figure}

We show that every typing derived by the \TFGJ rules gives rise to a
completion, that is, a well-typed FGJ program with the same structure.
\begin{definition}[Erasure]\label{def:erasure}
  Let $\mv{e}'$, $\mv{M}'$, $\mv{K}'$, $\mv{L}'$ be expression, method definition, constructor definition, class definition for FGJ. Define erasure functions 
  $\Erase{\mv{e}'}$, $\Erase{\mv{M}'}$, $\Erase{\mv{K}'}$,
  $\Erase{\mv{L}'}$ that map to the corresponding syntactic categories
  of \TFGJ as shown in Figure~\ref{fig:erasure}.
  \end{definition}
\begin{definition}[Completion]\label{def:completion}
  An FGJ expression $\mv{e}'$ is a \emph{completion} of a \TFGJ expression $\mv{e}$ if $\mv{e} = \Erase{\mv{e}'}$. Completions for method definitions, constructor definitions, and class definitions
  are defined analogously.
\end{definition}
\begin{theorem}
  Suppose that $\mathtt{\vdash \ol L : \Pi}$ such that $|\mathtt{\Pi (\exptype{C}{\ol{X} \triangleleft \ol{N}}.m)}| = 1$, for all $\mv{C.m}$ defined in $\ol L$. Then there is a completion $\ol{L}'$ of $\ol L$ such that
  $\ol{L}'\ \mathtt{OK}$ is derivable in FGJ.
\end{theorem}
\begin{proof}
  The proof is by induction on the length of $\ol L$.

  Consider the class typing $\mathtt{ \Pi \vdash \texttt{class}\ \exptype{C}{\ol{X} 
        \triangleleft \ol{N}} \triangleleft N\ \{ \ol{T} \
      \ol{f}; \ K \ \ol{M} \} \ \texttt{OK} : \Pi''}$ for an element of $\ol
    L$.

    We assume that all classes before $\mv{L}$ are completed according to the incoming $\mathtt{\Pi}$:
    If $\mathtt{\Pi (\exptype{D}{\ol{X} \triangleleft \ol{N}}.n)} = \mathtt{\exptype{}{\ol{Y} \triangleleft  \ol{P}}
      \ol{T} \to T}$, then $\mathtt{\exptype{}{\ol{Y} \triangleleft  \ol{P}}\ T\ \mv{n}(\ol{T}\
      \ol{x}) \dots}$ is in the completion of $\mv{D}$.

    Clearly, we can construct a completion for the class, if we can do so for each method. So we
    have to construct $\ol{M}'$ such that $\mathtt{\ol{M}'\ \mathtt{OK\ IN}\ \exptype{C}{\ol{X} 
        \triangleleft \ol{N}}}$. 

    Inversion of \rulename{GT-CLASS} yields
    \begin{gather}
      \label{eq:3}
      \mathtt{\Pi' = \Pi \cup \set{\exptype{C}{\ol{X} \triangleleft \ol{N}}.m \mapsto \exptype{}{} \ol{T_m} \to T_m \mid m \in \ol{M}} } \\
      \mathtt{\Pi'' = \Pi \cup \set{\exptype{C}{\ol{X} \triangleleft \ol{N}}.m \mapsto
          \exptype{}{\ol{Y} \triangleleft  \ol{P}} \ol{T_m} \to T_m \mid m \in \ol{M}} } \\
      \label{eq:5}
      \mathtt{\Pi', \Delta \vdash \ol{M} \ \texttt{OK IN}\
        \exptype{C}{\ol{X} \triangleleft \ol{N}} \extends N  \texttt{
          with } \exptype{}{\ol{Y} \triangleleft  \ol{P}}} \\
      \label{eq:7}
      \mathtt{\Delta = \ol{X} \subeq  \ol{N}, \ol{Y} \subeq  \ol{P}}
    \end{gather}
    Given some $\mv{M} = \texttt{m}(\ol{\mathtt{x}}) \{\texttt{return}\ \mathtt{e}_0;\} \in \ol{M}$,
    we show that
    \begin{gather}
      \label{eq:4}
      \exptype{}{\ol{Y} \triangleleft  \ol{P}}\ \mathtt{T_m}\ \texttt{m}(\ol{T_m}\ \ol{{x}})
      \{\texttt{return}\ \mathtt{e}'_0;\} \texttt{ OK IN }\exptype{C}{\ol{X} \triangleleft \ol{N}}
    \end{gather}
    is derivable for such completion $\mathtt{e_0'}$ of $\mathtt{e_0}$.

    By inversion of \eqref{eq:5} for $\mv{M}$, we obtain
    \begin{gather}
      \label{eq:6}
      \mathtt{\textit{override}(m, N, \exptype{}{\ol{Y} \triangleleft \ol{P}}\ol{T_m} \to T_m,
        \Pi) } \\
      \label{eq:9}
      \mathtt{\Pi; \Delta ; \ol{x}:\ol{T_m},\ this : \exptype{C}{\ol{X}} \vdash e_0 : S} \\
      \label{eq:8}
      \mathtt{\Delta \vdash S \subeq  T_m }
    \end{gather}
    As $\mathtt{\Delta}$ in~\eqref{eq:7} is defined as in \rulename{GT-METHOD'}, the well-formedness
    judgments are all given, the subtyping judgment \eqref{eq:8} is given as well as the override
    \eqref{eq:9}, the rule \rulename{GT-METHOD'} applies if we can establish
    \begin{gather}
      \label{eq:10}
      \mathtt{\Delta ; \ol{x}:\ol{T_m},\ this : \exptype{C}{\ol{X}} \vdash e_0' : S}
    \end{gather}
    for a completion of $\mathtt{e_0}$.

    To see that, we need to consider the rules \rulename{GT-NEW},
    \rulename{GT-CAST}, and \rulename{GT-INVK}. The 
    \rulename{GT-NEW} rule poses the existence of some $\ol{U}$ such that $\mv{N} =
    \exptype{C}{\ol{U}}$ for checking $\mv{e} = \mv{new}\ \mv{C} (\ol{e}) : \mv{N}$. In the completion, we
    define $\mv{e'} = \mv{new}\ \mv{N} (\ol{e}') : \mv{N}$ to apply rule \rulename{GT-NEW'} to the completions
    of the arguments.

    The rule \rulename{GT-CAST} splits into three rules
    \rulename{GT-UCAST'}, \rulename{GT-DCAST'}, and
    \rulename{GT-SCAST'}. These rules are disjoint, so that at most one of
    them applies to each occurrence of a cast. Here we assume a more
    liberal version of \rulename{GT-DCAST'} that admits downcasts that
    are not stable under type erasure semantics.

    For the rule \rulename{GT-INVK}, we first consider calls to methods not defined in the current
    class. By our assumption on previously checked classes $\mv{D}$ and their methods $\mv{n}$,
    $\mathit{mtype} (\mv{n}, \mv{D}, \mv\Pi) = \{\mathit{mtype}' (\mv{n}, \mv{D}')\}$ where the right
    side lookup happens in the completion following the definitions for FGJ (i.e., $\mv{D'}$ is the
    completion for $\mv D$). The \rulename{GT-INVK} rule poses the existence of some $\ol{V}$ that
    satisfies the same conditions as in \rulename{GT-INVK'}. Hence, we define the completion of
    $\mathtt{e_0.\mv{n}(\ol{e}) : [\ol{V}/\ol{Y}]U }$ as
    $\mathtt{e_0'.\exptype{n}{\ol{V}}(\ol{e}') : [\ol{V}/\ol{Y}]U }$.

    Next we consider calls to methods $\mv{n}$ defined in the current class, say, $\mv{C}$. For those methods,
    $\mathit{mtype} (\mv{n}, \mv{C}, \mv\Pi) = \exptype{}{} \ol{U} \to \mv U$, a non-generic
    type. By the definition of $\mathtt{\Pi''}$, we know that the type of this method will be
    published in the completion as $\exptype{}{\ol{Y} \triangleleft  \ol{P}} \ol{U} \to \mv
    U$. Hence, $\mathit{mtype}' (\mv{n}, \mv{C}') = \exptype{}{\ol{Y} \triangleleft  \ol{P}} \ol{U}
    \to \mv U$. As methods in $\mv{C}$ are mutually recursive, the rule must pose that $\ol{V} = \ol{Y}$ (cf.\
    Proposition~\ref{prop:polymorphi-recursion}). This setting fulfills all assumptions:
    \begin{gather}
      \label{eq:11}
      \mathtt{\Delta \vdash \ol{Y} \ \texttt{ok} } \\
      \mathtt{\Delta \vdash \ol{Y} \subeq  [\ol{Y}/\ol{Y}]\ol{P} }       
    \end{gather}
    We set the completion of
    $\mathtt{e_0.\mv{n}(\ol{e}) : [\ol{Y}/\ol{Y}]U }$ to
    $\mathtt{e_0'.\exptype{n}{\ol{Y}}(\ol{e}') : [\ol{Y}/\ol{Y}]U }$, which is derivable in FGJ.

    The remaining expression typing rules are shared between FGJ and \TFGJ, so they do not
    affect completions.
\end{proof}

\subsection{Polymorphic Recursion, Formally}
\label{sec:polym-recurs-form}

Consider an FGJ class $\mv{C}$ with $n$ mutually recursive methods $\mv{m}_i :
\forall\ol{X}_i. \ol{A}_i \to \ol{A}_i$, for $1\le i\le n$. Define the \emph{instantiation
  multigraph $IG(\mv{C})$} as a directed multigraph with vertices $\{1,
\dots, n\}$.
Edges between $i$ and $j$ in this graph are labeled with a
substitution from $\ol{X}_j$ to types in $\mv{m}_i$, which may contain
type variables from $\ol{X}_i$.
In particular, if $\mv{m}_i$ invokes $\mv{m}_j$ where the generic type variables in
the type of $\mv{m}_j$ are instantiated with substitution
$\ol{U}/\ol{X}_j$ (see rule GT-INVK), then 
$
i \stackrel{\ol{U}/\ol{X}_j}{\longrightarrow} j
$
is an edge of $IG (\mv{C})$.

Define the \emph{closure of the instantiation multigraph $IG^*(\mv{C})$} as the multigraph
obtained from $IG(\mv{C})$ by applying the following rule, which
composes the instantiating substitutions, exhaustively:
\begin{gather}\label{eq:1}
  i \stackrel{\ol{U}/\ol{X}_j}{\longrightarrow} j
  \quad\wedge\quad
  j \stackrel{\ol{V}/\ol{X}_k}{\longrightarrow} k
  \qquad\Rightarrow\qquad
  i \stackrel{[\ol{U}/\ol{X}_j]\ol{V}/\ol{X}_k}{\longrightarrow} k
\end{gather}

\begin{definition}\label{def:method-in-poly-rec}
  Method $\mv{m}_i$ is \emph{involved in polymorphic recursion}
  if there is an edge
  \begin{gather}\label{eq:2}
    i \stackrel{\ol{W}/\ol{X}_i}{\longrightarrow} i \quad \in IG^*
    (\mv{C}) \qquad \text{such that} \qquad \ol{W} \ne \ol{X}_i
  \end{gather}
\end{definition}
For the toy example in Figure~\ref{fig:examples-poly-rec}, we obtain
the multigraph $IG^* (\mv{UsePair})$ which indicates that \mv{prc} is
involved in polymorphic recursion:
\begin{gather*}
  \begin{array}{l@{\quad}|@{\quad}l}
    IG(\mv{UsePair})& IG^* (\mv{UsePair}) \\\hline
    \mv{prc} \stackrel{\mv{Y,X}/\mv{X,Y}}{\longrightarrow} \mv{prc} &
    \mv{prc} \stackrel{\mv{Y,X}/\mv{X,Y}}{\longrightarrow} \mv{prc} \qquad
    \mv{prc} \stackrel{\mv{X,Y}/\mv{X,Y}}{\longrightarrow} \mv{prc}
  \end{array}
\end{gather*}
The call to \mv{swap} does not appear in the graph because
\mv{swap} is defined in a different class.

For \mv{UsePair2}, we obtain a multigraph $IG^* (\mv{UsePair2})$ with
infinitely many edges which is also clear indication for polymorphic recursion:
\begin{gather*}
  \begin{array}{l@{\quad}|@{\quad}l}
    IG(\mv{UsePair2})& IG^* (\mv{UsePair2}) \\\hline
    \mv{prc} \stackrel{\mv{Y,Pair<X,Y>}/\mv{XY}}{\longrightarrow} \mv{prc}
    &
    \mv{prc} \stackrel{\mv{Y,Pair<X,Y>}/\mv{XY}}{\longrightarrow}
      \mv{prc} \\
    &
      \mv{prc}
      \stackrel{\mv{Pair<X,Y>,Pair<Y,Pair<X,Y>>}/\mv{XY}}{\longrightarrow}
      \mv{prc}
      \\
                     & \dots
  \end{array}
\end{gather*}

Clearly, $IG (\mv{C})$ is finite and can be constructed effectively by
collecting the instantiating substitutions from all method call
sites.
Repeated application of the propagation rule~\eqref{eq:1} either
results in saturation where no edge of the resulting multigraph satisfies~\eqref{eq:2} or it
detects an instantiating edge as in condition~\eqref{eq:2}. 

The following condition is necessary for the absence of
polymorphic recursion.

\begin{proposition}\label{prop:polymorphi-recursion}
  Suppose an FGJ class \mv{C} has $n$ methods, which are mutually
  recursive.
  If \mv{C} does not exhibit
  polymorphic recursion, then
  \begin{itemize}
  \item all methods quantify over the same number of generic
    variables;
  \item if a method has generic variables $\ol{X}$, then each call to
    a method of \mv{C} instantiates with a permutation of the
    $\ol{X}$;
  \item $IG^* (\mv{C})$ is finite.
  \end{itemize}
\end{proposition}
\begin{proof}
  Suppose for a contradiction that there are two distinct methods $\mv{m}_i$ and
  $\mv{m}_j$ with generic variables $\ol{X}_i$ and $\ol{X}_j$, respectively,
  where $|\ol{X}_i| < |\ol{X}_j|$. By mutual recursion, $\mv{m}_i$ invokes
  $\mv{m}_j$ directly or indirectly and vice versa. Hence, $IG^* (\mv{C})$
  contains edges from $i$ to $j$ and back:
  \begin{gather*}
    i \stackrel{\ol{U}/\ol{X}_j}{\longrightarrow} j
    \qquad
    j \stackrel{\ol{V}/\ol{X}_i}{\longrightarrow} i
  \end{gather*}
  As $IG^*(\mv{C})$ is closed under composition, it must also contain
  the edge
  \begin{gather*}
    j \stackrel{[\ol{V}/\ol{X}_i]\ol{U}/\ol{X}_j}{\longrightarrow} j
    \text{.}
  \end{gather*}
  By assumption $\mv{C}$ does not use polymorphic recursion, so it
  must be that $[\ol{V}/\ol{X}_i]\ol{U}/\ol{X}_j =
  \ol{X}_j/\ol{X}_j$. To fulfill this condition, all components of
  $\ol{U}$ must be variables $\in \ol{X}_i$. As  $|\ol{X}_i| <
  |\ol{X}_j| = |\ol{U}|$, there must be some variable $\mv{X} \in \ol{X}_i$ that
  occurs more than once in $\ol{U}$, say, at positions $j_1$ and $j_2$. 
  But that means the variables at positions $j_1$ and $j_2$ in
  $\ol{X}_j$ are mapped to the same component of $\ol{V}$. This is a
  contradiction because this substitution cannot be the identity
  substitution $\ol{X}_j/\ol{X}_j$.

  Hence, all methods have the same number of generic variables and all
  instantiations must use variables.

  Suppose now that there is a direct call from $\mv{m}_i$ to
  $\mv{m}_j$ where the instantiation $\ol{U}/\ol{X}_j$ is not a permutation. Hence,
  there is a variable that appears more than once in $\ol{U}$, which
  leads to a contradiction using similar reasoning as before.

  Hence, all instantiations must be permutations over a finite set of
  variables, so that $IG^*(\mv{C})$ is finite! 
\end{proof}

Moreover, if a class has only mutually recursive methods without
polymorphic recursion, we can assume that each method uses the same
generic variables, say $\ol{X}$, and each instantiation for
class-internal method calls is the identity $\ol X/\ol X$.

Using the same generic variables is achieved by $\alpha$ conversion.
By Proposition~\ref{prop:polymorphi-recursion}, we already know that
each instantiation is a permutation. Each self-recursive call must use
an identity instantiation already, otherwise it would constitute an
instance of polymorphic recursion. Suppose that method \mv{m} calls
method \mv{n} instantiated with a non-identity permutation, say
$\pi$ so that parameter $\mv{X_i}$ of \mv{n} gets instantiated with
$\mv{X_{\pi(i)}}$ of \mv{m}. In this case, we reorder the generic
parameters of \mv{n}  according to the inverse permutation
$\pi^{-1}$ and propagate this permutation to all call sites of
\mv{n}. For the call in \mv{m}, we obtain the identity permutation
$\pi \cdot \pi^{-1}$, for self-recursive calls inside \mv{n}, the
instantiation remains the identity (for the same reason), for a
call site in another method which instantiates \mv{n} with permutation
$\sigma$, we change that permutation to $\sigma \cdot \pi^{-1}$, which
is again a permutation.
This way, we can eliminate all non-identity instantiations from calls
inside \mv{m}.

We move our attention to \mv{n}. Each self-recursive call and each call to \mv{m} uses the
identity instantiation, the latter by construction. So we only need to
consider calls to $\mv{n'}\notin\{\mv{n}, \mv{m}\}$ with an
instantiation which is not the identity permutation. We can also
assume that $\mv{n'}$ is not called from $\mv{m}$: otherwise, \mv{n'}
would have the generic variables in the same order as \mv{m} and hence
as \mv{n}! But that means we can fix all calls to $\mv{n'}$ by
applying the inverse permutations as for \mv{n} \emph{without disturbing the already
  established identity instantiations}!

Each such step eliminates all non-identity instantiations for at least
one method without disturbing previous identity instantiations. Hence,
the procedure terminates after finitely many steps with a class with
all instantiations being identity permutations.

\section{Type inference algorithm}
\label{sec:type-infer-algor}
This section presents our type inference algorithm.
The algorithm is given method assumptions $\mv\Pi$ and applied to a
single class $\mv L$ at a time:
\begin{gather*}
\fjtypeinference(\mv{\Pi}, \texttt{class}\ \exptype{C}{\ol{X}
\triangleleft \ol{N}} \triangleleft \type{N}\ \{ \ldots \}) = \\
\quad \quad \begin{array}[t]{rll}
  \textbf{let}\ 
  (\overline{\methodAssumption}, \consSet) &= \fjtype{}(\mv{\Pi}, \texttt{class}\ \exptype{C}{\ol{X}
  \triangleleft \ol{N}} \triangleleft \type{N}\ \{ \ldots \}) &
                                                                     \text{// constraint generation}\\
              {(\sigma,  \ol{Y} \triangleleft \ol{P})} &= \unify{}(\consSet,\, \ol{X} <: \ol{N}) & \text{// constraint solving}\\
\end{array}\\
\textbf{in}\ \mv\Pi \cup
\set{(\exptype{C}{\ol{X} \triangleleft \ol{N}}.\mv{m} : \exptype{}{\ol{Y} \triangleleft \ol{P}}\ \ol{\sigma(\tv{a})} \to \sigma(\tv{a})) \ |\ (\exptype{C}{\ol{X} \triangleleft \ol{N}}.\mv{m} : \ol{\tv{a}} \to \tv{a}) \in \overline{\methodAssumption}}
\end{gather*}

The overall algorithm is nondeterministic. The function $\unify{}$ may
return finitely many times as there may be multiple solutions for a constraint
set.  A local solution for class $\mv C$ may not
be compatible with the constraints generated for a subsequent class. In this case, we have to backtrack to $\mv C$ and proceed to the next
local solution; if thats fail we have to backtrack further to an earlier class.

\subsection{Type inference for a program }
\label{sec:proc-mult-class}

{Type inference processes a program one class at a time.
To do so, it must be possible to order the classes such that early
classes never call methods in later classes.
As an example, Figure~\ref{fig:invalidinput} shows a program that is
acceptable in FGJ, but rejected by \TFGJ because the methods \mv{m1}
and \mv{m2} are mutually recursive across class boundaries. There is
no order in which classes \mv{C1} and \mv{C2} can be processed.

Figure~\ref{fig:correctinput} contains a program
acceptable to both \TFGJ and FGJ because the mutual recursion of
methods \mv{m1} and \mv{m2} is taking place inside class \mv{D2}. As
\mv{D2} invokes method \mv{m} of \mv{D1}, type inference must process
\mv{D1} before \mv{D2}, which corresponds to the constraints imposed
by the typing of \TFGJ in Section~\ref{chapter:type-rules}.}

\begin{figure}[tp]
    \centering
    \begin{minipage}{.48\textwidth}
\begin{lstlisting}[style=tfgj]
class C1 extends Object {
  m1(){ return new C2().m2(); }
}
class C2 extends Object{
  m2(){ return new C1().m1(); }
}
\end{lstlisting}
      \caption{Invalid \TFGJ program}
      \label{fig:invalidinput}
    \end{minipage}%
  ~$\left|~
    \begin{minipage}{.48\textwidth}
\begin{lstlisting}[style=tfgj]
class D1 extends Object {
  m(){ return ...; }
}
class D2 extends Object{
  m1(x){ return new D2().m2(); }
  m2( ){ return new D2().m1(
                new D1().m()); }
}
\end{lstlisting}
      \caption{Valid \TFGJ program}
      \label{fig:correctinput}
    \end{minipage}\right.$
\end{figure}

We obtain a viable order for processing the class declarations by
computing an approximate call graph based solely on method names. That
is, if method \mv{m} is used in \mv{C3} and defined both in \mv{C1}
and \mv{C2}, then \mv{C1} and \mv{C2} must both be processed before
\mv{C3}. In such a case, the use of \mv{m} might be ambiguous so that
type inference for class \mv{C3} proposes more than one solution. Global
type inference attempts to extend each partial solution to a solution for the
whole program and backtracks if that fails.

\subsection{Constraint generation}
\label{sec:fjtype}
\begin{figure}[tp]
  \begin{align*}
    \itype{T}, \itype{U} &::= \tv{a} \mid \mv{X} \mid {\itype{N}} && \text{type variable, bounded type parameter, or type}\\
    \itype{N} &::= \exptype{C}{\il{T}} && \text{class type (with type variables)} \\
    \simpleCons &::= \itype{T} \lessdot \itype{U} \mid \itype{T} \doteq \itype{U} && \text{simple constraint: subtype or equality}\\
    \orCons{} &::= \set{\set{\overline{\simpleCons_1}}, \ldots, \set{\overline{\simpleCons_n}}} && \text{or-constraint}\\
    \constraint &::= \simpleCons \mid \orCons && \text{constraint}\\
    \consSet &::= \set{\constraints} && \text{constraint set}\\
    \methodAssumption &::= \exptype{C}{\ol{X} \triangleleft \ol{N}}.\texttt{m} : \exptype{}{\ol{Y}
                        \triangleleft \ol{P}}\ \ol{\itype{T}} \to \itype{T}  &&
                                                                \text{method
                                                                type assumption}\\
    \localVarAssumption &::= \texttt{x} : \itype{T} && \text{parameter
                                                       assumption}\\
    \mtypeEnvironment & ::= \mv\Pi \cup \overline{\methodAssumption} &
                & \text{method type environment} \\
    \typeAssumptionsSymbol &::= ({\mtypeEnvironment} ; \overline{\localVarAssumption}) 
  \end{align*}
  \caption{Syntax of constraints and type assumptions}
  \label{fig:syntax-constraints}
\end{figure}

Figure~\ref{fig:syntax-constraints} defines the syntax of
constraints. We extend types with \emph{type variables} ranged over by
$\tv{a}$. A constraint is either a simple constraint $\simpleCons$ or
an or-constraint $\orCons$, which is a set of sets of simple
constraints. An or-constraint represents different alternatives,
similar to an intersection type, and
cannot be nested. The output of constraint generation is a set of
constraints $\consSet$, which can hold simple constraints as well as or-constraints.

Figure~\ref{fig:constraints-for-classes} contains the algorithm
{\fjtype} to generate constraints for classes. Its input
consists of the method type environment $\mv\Pi$ of the previously
checked classes. It distinguishes between overriding and
non-overriding method definitions. The former are recognized by
successful lookup of their type using \textit{mtype}. We set up the
method type assumptions accordingly and generate a constraint between
the inferred return type $\tv{a}_{\mv m}$ and the one of the
overridden method to allow for covariant overriding.
Constraints for the latter methods are generated with all fresh type
variables for the argument and result types.

\begin{figure}[tp]
  \begin{gather*}
    \begin{array}{@{}l@{}l}
      \fjtype & ({\mv{\Pi}}, \mathtt{class } \ \exptype{C}{\ol{X} \triangleleft \ol{N}} \ \mathtt{ extends } \ \mathtt{N \{ \overline{T} \ \overline{f}; \, K \, \overline{M} \}}) =\\
              & \begin{array}{ll@{}l}
                  \textbf{let} & \overline{\tv{a}}_{\mv m} \text{ be fresh type variables
                                 for each }\mv{m}\in\ol{M} \\
                               & \ol{\methodAssumption}_o = \set{\exptype{C}{\ol{X} \triangleleft \ol{N}}.\mv{m} :
                                 \exptype{}{\ol{Y} \triangleleft \ol{P}}\ \ol{T} \to \tv{a}_{\mv m} \mid \mv{m} \in \ol{M},
                                 \textit{mtype}(\mv{m}, \type{N}, \mv\Pi) = \exptype{}{\ol{Y} \triangleleft \ol{P}} \ol{\type{T}} \to \type{T}}\\
                               & \consSet_o = \set{\tv{a}_{\mv m} \lessdot \type{T} \ |\ \texttt{m} \in \ol{M}, 
                                 \textit{mtype}(\mv{m}, \mv{N}, \mv\Pi) = \exptype{}{\ol{Y} \triangleleft \ol{P}} \ol{\type{T}} \to \type{T}}\\
                               & \ol{\methodAssumption}' = 
                                 \set{ (\exptype{C}{\ol{X} \triangleleft \ol{N}}.\mv{m} : \ol{\tv{a}} \to \tv{a}_{\mv
                                 m}) \mid \mv{m} \in \ol{M},
                                 \textit{mtype} (\mv m, \mv N,
                                 \mv\Pi) \text{ not defined}, \ol{\tv{a}}\ \text{fresh} } \\
                               & \consSet_m = \set{ \set{\tv{a}_{\mv m} \lessdot \type{Object}, \, \ol{\tv{a}} \lessdot \ol{\type{Object}}}
                                \ |\ (\exptype{C}{\ol{X} \triangleleft \ol{N}}.\mv{m} : \ol{\tv{a}} \to \tv{a}_{\mv m}) \in \ol{\methodAssumption}'} \\
                               & \Pi = \mv{\Pi} \cup
                                 \ol{\methodAssumption}' \cup
                                 \ol{\methodAssumption}_o \\
                  \textbf{in} 
                               & { ( \Pi, 
                                 \consSet_o \cup \consSet_m \cup
                                 \bigcup_{\texttt{m} \in \ol{M}}
                                 \typeMethod(\Pi, \exptype{C}{\ol{X}},  \mathtt{m}))
                                 } 
                \end{array}
    \end{array}
  \end{gather*}
  \caption{Constraint generation for classes}
  \label{fig:constraints-for-classes}
\end{figure}

Constraint generation alternates with constraint solving: After
generating constraints with {\fjtype}, we solve them to obtain one or
more candidate extensions for the method type environment
$\mv\Pi$. Next, we pick a candidate and continue with the next class
until all classes are checked and we have an overall method type
environment.  Otherwise, we backtrack to check the next candidate. 

\begin{gather*}
\begin{array}{@{}l@{}l@{}l}
  \typeMethod & (\ensuremath{\mv{\Pi} }, \exptype{C}{\ol{X}},\ &\  \mathtt{m}(\ol{x})\{ \mathtt{ return }\ \texttt{e}; \}) =\\
              & \textbf{let} & \exptype{}{\ol{Y} \triangleleft \ol{P}}\ \overline{T} \to T  = \mv{\Pi} (\exptype{C}{\ol{X} \triangleleft \ol{N}}.\mv{m})
  \\
              & 
                                                                                              & ({R}, \consSet) =
                                                                                                \typeExpr(\mv\Pi ; \set{\mv{this} :
                                                                                                \exptype{C}{\ol{X}} } \cup \set{ \ol{x} : \overline{{T}} }), \texttt{e})\\
              & \mathbf{in}
                                                                                              & \consSet \cup \set{{R} \lessdot {T}}\\
\end{array}
\end{gather*}

The \typeMethod{} function for methods calls the \typeExpr{} function with the
return expression. It adds the assumptions for \texttt{this} and for the method parameters to the global assumptions before passing them to \typeExpr.

\smallskip

In the following we define the \typeExpr{} function for every possible expression:

\smallskip

\noindent
$\typeExpr: \typeAssumptionsSymbol \times
\texttt{Expression} \rightarrow {T} \times \consSet
$
\begin{gather*}
  \typeExpr (({\mtypeEnvironment} ; 
  \overline{\localVarAssumption}), \mathtt{x}) =
  (\overline{\localVarAssumption} (\mv x) , \emptyset)
\end{gather*}
When we encounter a field $\mv e.\mv f$, we consider all classes
$\mv C$
that define field $\mv f$ and impose an or-constraint that covers all
alternatives: the type $R$ of the expression $\mv e$ must be a subtype of a generic instance
of $\mv C$ and the return type must be the corresponding field type. 
\begin{gather*}
  \begin{array}{@{}l@{}l}
    \typeExpr{} &(({\mtypeEnvironment} ;
                  \overline{\localVarAssumption}), \texttt{e}.\texttt{f}) = \\
                & \begin{array}{ll}
                    \textbf{let} 
                    & ({R}, \consSet_R) = \typeExpr(({\mtypeEnvironment} ;
                      \overline{\localVarAssumption}), \texttt{e})\\
                    & \tv{a} \text{ fresh} \\
                    & \constraint = \begin{array}[t]{@{}l@{}l}
                      \orCons\set{
                      \set{ &
                      {R} \lessdot \exptype{C}{\ol{\tv{a}}} , \tv{a} \doteq
                      [\overline{\tv{a}}/\ol{X}]\type{T} , \ol{\tv{a}} \lessdot [\overline{\tv{a}}/\ol{X}]\ol{N}
                      \mid \overline{\tv{a}} \text{ fresh}
                      } \\
                      & \quad \mid \mv{T}\ \mv{f} \in \texttt{class}\ \exptype{C}{\ol{X} \triangleleft \ol{N}} \set{ \ol{T}\ \ol{f}; \ldots}
                      }\end{array}\\
                    {\mathbf{in}} & {
                    (\tv{a}, (\consSet_R \cup \set{\constraint}))}
                  \end{array} 
  \end{array}
\end{gather*}
We treat method calls in a similar way. We impose an or-constraint
that considers a generic instance of a method type in a class
providing that method (with the same number of parameters). Each
choice imposes a subtyping constraint on the receiver type  $R$ as
well as subtyping constraints on the argument types
$\overline{R}$. Moreover, we need to check that the subtyping
constraints of the method type are obeyed by instantiating them
accordingly.

\noindent
$\begin{array}{@{}l@{}l}
\typeExpr{} & (({\mtypeEnvironment} ;
  \overline{\localVarAssumption}), \texttt{e}.\mathtt{m}(\overline{\texttt{e}}) ) = \\
& \begin{array}{ll}
\textbf{let}& ({R}, \consSet_R) = \typeExpr(({\mtypeEnvironment} ;
  \overline{\localVarAssumption}), \texttt{e})\\
& \forall \texttt{e}_i \in \ol{e} : (R_i, \consSet_i) = \typeExpr(({\mtypeEnvironment} ;
  \overline{\localVarAssumption}), \texttt{e}_i)  \\
                    & \tv{a} \text{ fresh} \\
& \begin{array}{@{}l@{}l}
  \constraint = \orCons\set{ & \{
            \begin{array}[t]{l}
              {R} \lessdot \exptype{C}{\ol{\tv{a}}}, \tv{a} \doteq [\ol{\tv{b}}/\ol{Y}][\ol{\tv{a}}/\ol{X}]{T},
              \overline{R} \lessdot [\ol{\tv{b}}/\ol{Y}][\ol{\tv{a}}/\ol{X}]\overline{T} , \\
              \ol{\tv{b}} \lessdot [\ol{\tv{b}}/\ol{Y}][\ol{\tv{a}}/\ol{X}]\ol{P} ,
              \ol{\tv{a}} \lessdot [\ol{\tv{a}}/\ol{X}]\ol{N}
              \mid \overline{\tv{a}}, \overline{\tv{b}} \text{ fresh} \}
            \end{array}\\
      & \ |\ (\exptype{C}{\ol{X} \triangleleft \ol{N}}.\texttt{m} : \exptype{}{\ol{Y} \triangleleft \ol{P}} \ \overline{T} \to {T}) \in {\mtypeEnvironment} }
  \end{array}\\
\mathbf{in} & (\tv{a},(\consSet_R \cup \bigcup_i \consSet_i \cup \set{\constraint}))
\end{array}
\end{array}
$

The \texttt{new}-expression is comparatively simple. Starting from a
generic instance of the class type, we calculate the types $\overline
T$ of the fields, impose subtyping constraints on the constructor
argument $\overline R$, and check the subtyping  constraints of the class.
\begin{gather*}
  \begin{array}{@{}l@{}l}
    \typeExpr &(({\mtypeEnvironment} ; 
                \overline{\localVarAssumption}), \mathtt{new }\ \texttt{C}(\ol{e}) ) = \\
              &\begin{array}{ll}
                 \textbf{let} 
                 & \forall \texttt{e}_i \in \overline{\texttt{e}} : (R_i, \consSet_i) = \typeExpr(({\mtypeEnvironment} ;
                   \overline{\localVarAssumption}), \texttt{e}_i)  \\
                 & \overline{\tv{a}} \text{ fresh} \\
                 & \textit{fields} (\exptype{C}{\overline{\tv{a}}}) = \overline{T}\ \ol{f} \\
                 & \consSet = \set{\overline{R} \lessdot \overline{T}}
                   \cup \set{\overline{\tv{a}} \lessdot
                   [\overline{\tv{a}}/\ol{X}]\ol{N}} 
                   \qquad\text{where}\ \texttt{class}\ \exptype{C}{\ol{X} \triangleleft \ol{N}} \set{\ldots} \\
                 \mathbf{in}& (\exptype{C}{\ol{\tv{a}}}, \consSet \cup \bigcup_i \consSet_i)
               \end{array}
  \end{array}
\end{gather*}

For cast expressions, we ignore the return type and pass on the
constraints for the subexpression. We return the target type of the cast.
\begin{gather*}
  \begin{array}{@{}l@{}l}
    \typeExpr &(({\mtypeEnvironment} ;
                \overline{\localVarAssumption}), (\texttt{N})\mv{e} ) = \\
              &\begin{array}{ll}
                 \textbf{let} 
                 & (R, \consSet) = \typeExpr(({\mtypeEnvironment} ;
                   \overline{\localVarAssumption}), \texttt{e})  \\
                 \mathbf{in} & (\mv{N}, \consSet)
               \end{array}
  \end{array}
\end{gather*}

\begin{example}\label{example:typeExpr}
To illustrate the constraint generation step we will apply it to the program depicted in
figure \ref{fig:intro-example-generic-jtx}.
First the \fjtype{} function assigns the fresh type variable $\tv{f}$ to the parameter \texttt{fst}.
Afterwards the \typeExpr{} function is called on the return expression of the \texttt{setfst} method.
The local variable \texttt{fst} does not emit any constraints.
For the \texttt{this.snd} part of the expression the \typeExpr{} function returns an or-constraint:\\
\noindent
$\begin{array}{ll}
\constraints_{1} &= \typeExpr(({\mtypeEnvironment} ;
\overline{\localVarAssumption}), \texttt{this.snd}) \\
&= (\tv{b}, \orCons(\set{
(\exptype{Pair}{\type{X}, \type{Y}} \lessdot \exptype{Pair}{\tv{w}, \tv{y}}),
(\tv{b} \doteq \tv{y}), (\tv{w} \lessdot \type{Object}), (\tv{y} \lessdot \type{Object}) 
}) )
\end{array}$\\
\noindent
This constraint is merged with the constraints generated by the \texttt{new Pair} constructor call:
\\
\noindent
$\begin{array}{l}
\typeExpr(({\mtypeEnvironment} ;
  \overline{\localVarAssumption}), \texttt{new Pair(fst, this.snd)})\\
  \quad \quad = (\exptype{Pair}{\tv{d},\tv{e}}, \set{
  (\tv{f} \lessdot \tv{d}), (\tv{b} \lessdot \tv{e}),
  (\tv{d} \lessdot \texttt{Object}), (\tv{e} \lessdot \texttt{Object})} \cup \constraints_1)
\end{array}$
\end{example}

\section{Constraint Solving}
\label{sec:unify}

This section describes the \unify{} algorithm
which is used to find solutions for the constraints generated by \fjtype{}.

It first attempts to transforms a constraint set into solved form and
reads off a solution in the form of a substitution.

\begin{definition}[Solved form]\label{def:solved-form}
  A set $C$ of constraints is in solved form if it only contains
  constraints of  the following form:
  \begin{enumerate}
  \item $\tv{a} \lessdot \tv{b}$ 
  \item $\tv{a} \doteq \tv{b}$,
  \item\label{item:1} $\tv{a} \lessdot \exptype{C}{\ol{\itype{T}}}$, 
  \item\label{item:2} $\tv{a} \doteq \exptype{C}{\ol{\itype{T}}}$, with $\tv{a} \notin \ol{\itype{T}}$.
  \end{enumerate}
  In case~\ref{item:1} and~\ref{item:2} the type variable $\tv{a}$ does not appear on the left of another constraint of the
  form~\ref{item:1} or~\ref{item:2}.
\end{definition}  

For brevity, we write $\tv{a}_0 \lessdot^* \tv{a}_n$ for a non-empty chain of subtyping constraints between type variables $\tv{a}_0 \lessdot
\tv{a}_1, \tv{a}_1 \lessdot \tv{a}_2, \dots, \tv{a}_{n-1} \lessdot \tv{a}_n$ where $n>0$.

\subsection{{Algorithm $\unify{} (C, \Delta)$}}


The input of the algorithm is a set of constraints $\consSet$ and a
type environment $\Delta$. The type environment binds the generic type
variables $\ol{X}$ to their upper bounds. It is used in invocations of
the subtyping judgment.

The treatment of the generic class variables $\ol{X} \extends \ol{N}$
deserves some explanation. The algorithm must not substitute
for these variables. Instead it treats them like parameterless
abstract classes $\exptype{X_i}{}$ which are subtypes of their
respective $\mv{N}_i$ (where the variable name $\mv{X}_i$ is now
treated like a class name). Example~\ref{item:3} illustrates this
approach.  

The first step of the algorithm eliminates or-constraints from constraint set $\consSet$. To do so,
we consider all combinations of selecting simple constraints from
or-constraints in $\consSet$. In general, we have that $\consSet =
\set{\overline{\simpleCons}, \orCons_1, \dots, \orCons_n}$ and we
execute the remaining steps for all $\consSet' = \set{\overline{\simpleCons}}
\cup \set{\overline{\simpleCons}_1} \cup \dots \cup
\set{\overline{\simpleCons}_n}$ where $\overline{\simpleCons}_i \in
\orCons_i$.

\textbf{Step 1.} We apply the rules in Figures~\ref{fig:fgjreduce-rules}
and~\ref{fig:fgjerase-rules} exhaustively to $C'$.

\textbf{Step 2.} At this point, all
constraints $ \simpleCons \in C'$ are either in solved form or one of the following
cases applies:
\begin{enumerate}
\item $\set{\exptype{C}{\il{T}} \lessdot \exptype{D}{\il{U}}}
  \subseteq C'$ where $\forall\ol{X},\ol{N}\colon \Delta \not\vdash \exptype{C}{\ol{X}} <:
  \exptype{D}{\ol{N}}$ (roughly, $\mv C$ cannot be a subtype of $\mv D$) --- in this case $C'$ has no solution;
\item $\set{\tv{a} \lessdot \exptype{C}{\il{T}}, \tv{a} \lessdot
    \exptype{D}{\il{V}}} \subseteq C'$ where
  $\forall\ol{X},\ol{N}\colon \Delta \not\vdash \exptype{C}{\ol{X}} <:
  \exptype{D}{\ol{N}} $ and $\forall\ol{X},\ol{N}\colon \Delta \not\vdash \exptype{D}{\ol{X}} <:
  \exptype{C}{\ol{N}}$ (roughly, $\mv C$ and $\mv D$ are not
  subtype-related) --- in this case $C'$ has no solution; or
\item $\set{\exptype{C}{\il{T}} \lessdot \tv{b}} \subseteq C'$.
\end{enumerate}
The last case is a lower bound constraint which is embraced by Scala,
but which is not legal in FGJ (nor in Java). As we insist on inferring
a type, we have to find a concrete instance for
$\exptype{C}{\il{T}}$. To do so, we generate an or-constraint from
each lower bound constraint and its corresponding upper bound constraint (using upper bound \mv{Object} if no such constraint exists) as follows:
\begin{align*}
  \expandLB (\exptype{C}{\il{T}} \lessdot \tv{b}, \tv{b} \lessdot \exptype{D}{\il{U}})
  & = \set{ \set{ \tv{b} \doteq [\overline{\itype{T}}/\ol{X}]\type{N} } \mid
    \Delta \vdash \exptype{C}{\overline{\type{X}}} <: \type{N},
    \Delta \vdash \type{N} <:  \exptype{D}{\overline{P}}}
  \\ \text{where }&\ol{P}\text{ is determined by }\Delta \vdash \exptype{C}{\overline{\type{X}}} <: \exptype{D}{\ol{P}} \text{ and
  } [\overline{\itype{T}}/\ol{X}]\ol{P} = \il{U}
\end{align*}
This constraint replaces the lower and upper bound constraint from which it was generated.

A lower bound may also be implied by a constraint set with constraints of the form
$C_{ab} = \tv{a} \lessdot \exptype{C}{\il{T}}, \tv{a} \lessdot^* \tv{b}$.
In this case $\exptype{C}{\il{T}}$ must either be a upper or lower bound for $\tv{b}$.
We implement it by $\expandLB$, which adds a lower bound constraint for $\tv{b}$
and also adding a upper bound to $\tv{b}$. While $C_{ab}$ remains in the constraint
set:
$\expandLB (\exptype{C}{\il{T}} \lessdot \tv{b}, \tv{b} \lessdot \exptype{D}{\il{U}})\ \cup \ 
\set{b \lessdot \exptype{C}{\il{T}}}$

Now we are in a similar situation as before. Our current constraint set $C'$ is a mix of simple constraints and or-constraints and, again, we
consider all (simple) constraint sets $C''$ that arise as combinations of selecting simple constraints from $C'$.

\textbf{Step 3.} We apply the rule (subst) exhaustively to $C''$:
\begin{gather*}
  \begin{array}[c]{lll}
    \text{(subst)} &
                     \begin{array}[c]{l}
                       C \cup \set{\tv{a} \doteq \itype{T}}\\
                       \hline
                       [\itype{T}/\tv{a}]C \cup \set{\tv{a} \doteq \itype{T}}
                     \end{array}
                   & \tv{a} \text{ occurs in } C \text{ but not in } \itype{T}
  \end{array}
\end{gather*}
We fail if we find any $\tv{a} \doteq \itype{T}$ such that $\tv{a}$ occurs in $\itype{T}$.

\textbf{Step 4.} If $C''$ has changed from applying (subst), we continue with $C''$ from step~1.

\textbf{Step 5.} Otherwise, $C''$ is in solved form and it remains to eliminate subtyping constraints between variables by exhaustive application of rule
(sub-elim) and (erase) (see Figure~\ref{fig:fgjerase-rules}). Applying this rule does not affect the solve form property.
\begin{gather*}
  \begin{array}[c]{lll}
    \text{(sub-elim)} &
                     \begin{array}[c]{l}
                       C \cup \set{\tv{a} \lessdot \tv{b}}\\
                       \hline
                       [\tv{a}/\tv{b}]C \cup \set{\tv{b} \doteq \tv{a}}
                     \end{array}
  \end{array}
\end{gather*}
\textbf{Step 6.} We finish by generating a solving substitution from the remaining $\doteq$-constraints
and generic variable declarations from the remaining $\lessdot$-constraints. Let $C'' = C_{\doteq} \cup C_{\lessdot}$ such that $C_{\doteq}$
contains only $\doteq$-constraints and $C_{\lessdot}$ contains only $\lessdot$-constraints. Now $C_{\lessdot} = \set{ \il{\tv{a}} \lessdot
  \il{N}}$ and choose some fresh generic variables $\ol{Y}$ of the same length as $\il{\tv{a}}$. We can read off the substitution $\sigma$
from $C_{\doteq}$ where we need to substitute the generic variables for the type variables. We obtain the generic variable declarations
directly from $C_{\lessdot}$ using the same generic variable substitution. We need not apply $\sigma$ here because
we applied (subst) exhaustively in Step~3.
\begin{align*}
  \sigma &=
           \set{\tv{b} \mapsto [\ol{Y}/\ol{\tv{a}}]\itype{T} \mid
           (\tv{b} \doteq \itype{T}) \in C_{\doteq}}
           \cup 
           \set{\il{\tv{a}} \mapsto \ol{Y}}
           \cup 
           \set{{\tv{b}} \mapsto \mv{X} \mid (\tv{b} \lessdot \mv{X}) \in C_{\lessdot}}, \\
  \unifyGenerics{} & =
                     \set{\type{Y} \triangleleft [\ol{Y}/\ol{\tv{a}}]\itype{N} \mid (\tv{a} \lessdot \itype{N}) \in C_{\lessdot}}
\end{align*}
We return the pair $(\sigma, \gamma)$.

\begin{figure}
\begin{center}
    \leavevmode
    \fbox{
    \begin{tabular}[t]{l@{~}l}
      (match)
      & $
      \begin{array}[c]{@{}ll}
      \begin{array}[c]{l}
         C \cup \, \set{\tv{a} \lessdot
         \exptype{C}{\il{T}},
         \tv{a} \lessdot
          \exptype{D}{\il{V}}} \\ 
        \hline
        \vspace*{-0.4cm}\\
        C \cup \set{\tv{a} \lessdot \exptype{C}{\il{T}}
        , \exptype{C}{\il{T}} \lessdot \exptype{D}{\il{V}}}
      \end{array}
      & 
        \Delta \vdash \exptype{C}{\ol{X}} <: \exptype{D}{\ol{N}} 
      \end{array}
      $
    \\\\
    (adopt)
    & $
    \begin{array}[c]{@{}ll}
    \begin{array}[c]{l}
       C \cup \, \set{\tv{a} \lessdot
       \exptype{C}{\il{T}},
       \tv{b} \lessdot^* \tv{a}, \tv{b} \lessdot \exptype{D}{\il{U}}} \\ 
      \hline
      \vspace*{-0.4cm}\\
      C \cup \set{
        \tv{a} \lessdot
       \exptype{C}{\il{T}},
       \tv{b} \lessdot^* \tv{a}
      , \tv{b} \lessdot  \exptype{D}{\il{U}}
      , \tv{b} \lessdot \exptype{C}{\il{T}}
       }
    \end{array}
    \end{array}
    $
  \\\\
      (adapt)
      & $
      \begin{array}[c]{@{}ll}
      \begin{array}[c]{l}
         C \cup \, \set{\exptype{C}{\il{T}} \lessdot
          \exptype{D}{\il{U}}} \\ 
        \hline
        \vspace*{-0.4cm}\\
        C \cup \set{\exptype{D}{[ \il{T} / \ol{X} ]\ol{N}}
        \doteq \exptype{D}{\il{U}}}
      \end{array}
        & 
          \Delta \vdash \exptype{C}{\ol{X}} <: \exptype{D}{\ol{N}} 
      \end{array}
      $
    \\\\
(reduce) & $
\begin{array}[c]{l}
  C \cup \set{\exptype{D}{\il{T}} \doteq
    \exptype{D}{\il{U}}}\\
  \hline
  C \cup \set{\il{T} \doteq \il{U}}
\end{array}
      $ \\\\
(equals) & $
\begin{array}[c]{l}
  C \cup \set{\tv{a}_1 \lessdot
  \tv{a}_2, \tv{a}_2 \lessdot \tv{a}_3, \dots, \tv{a}_n \lessdot \tv{a}_1}\\
  \hline
  C \cup \set{\tv{a}_1 \doteq \tv{a}_2, \tv{a}_2 \doteq \tv{a}_3, \dots}
\end{array} \quad n>0
      $ \\\\
    \end{tabular}}
  \end{center}
\caption{Reduce and adapt rules}\label{fig:fgjreduce-rules}
\end{figure}

\begin{figure}
\begin{center}
\fbox{\begin{tabular}[t]{ll}
      (erase)  & $ 
      \begin{array}[c]{@{}ll}
        \begin{array}[c]{l}
          C \cup \set{\tv{a} \doteq \tv{a}}\\
          \hline
          C
        \end{array}
      \end{array}$
  \\
       (swap) & $
            \begin{array}[c]{@{}ll}
              \begin{array}[c]{l}
                C \cup \set{\itype{N} \doteq \tv{a}}\\
                \hline
                C \cup \set{\tv{a} \doteq \itype{N}}
              \end{array}
            \end{array}$
          \end{tabular}}
        \end{center}
\caption{Erase and swap rules}\label{fig:fgjerase-rules}
\end{figure}

\begin{example}\label{item:3}
  To illustrate our treatment of generic variables, we consider a
  typical case involving the (adapt) rule from
  Figure~\ref{fig:fgjreduce-rules}.

  Consider $C = \set{ \mv{X} \lessdot \exptype{D}{\il{U}} }$ and let $
  \mv{X} \sub \exptype{C}{\ol{T}} \in \Delta$ be the bound for
  $\mv{X}$.

  The side condition of the rule (adapt) asks for some $\ol{N}$ such
  that
  $\Delta \vdash \mv{X} \sub \exptype{D}{\ol{N}}$, i.e., ``is there a
  way that $\mv{X}$ can be a subtype of $\mv{D}$?''

  By inversion of subtyping and transitivity, this judgment holds if
  $\Delta \vdash \exptype{C}{\ol{T}} \sub \exptype{D}{\ol{N}}$ holds.

  Hence, applying (adapt) to $C$ yields $\set{ \exptype{D}{\ol{N}} \doteq
    \exptype{D}{\il{U}} }$. The substitution in the rule is empty
  because $\mv{X}$ is considered a parameterless type.

  The remaining rules work similarly. In particular, different
  variables $\mv{X} \ne \mv{Y}$ give rise to different (abstract)
  classes. For example, the (reduce) rule removes the constraint $\mv{X} \doteq
  \mv{X}$, but it does \emph{not} apply to
  $\mv{X} \doteq \mv{Y}$. Rather, an equation like this renders the
  constraint set unsolvable.
\end{example}


\subsection{Properties of \unify{}}
\label{sec:properties-unify}

First we give some definitions and results.
For the complete proofs see appendix \ref{chapter:unifySoundnessProof}, \ref{chapter:unifyCompletenessProof} and \ref{chapter:unifyTerminationProof}.

\begin{definition}[Unifier]
  Let $C$ be a set of constraints and $\Delta$ a type environment.
  A substitution $\sigma$  is a \emph{unifier} of $(C,\Delta)$ if
  \begin{itemize}
  \item for each $(\itype{T} \lessdot \itype{U}) \in C$ it holds that
    $\Delta \vdash \exp{\sigma}{T} \olsub \exp{\sigma}{U}$;
  \item  for each $(T \doteq U) \in C$ it holds that
    $\exp{\sigma}{T} = \exp{\sigma}{U}$; and
  \item for each or-constraint $\set{\set{\overline{\simpleCons_1}},
      \ldots, \set{\overline{\simpleCons_n}}} \in C$, there exists $1
    \le i \le n$ such that $\sigma$ is a unifier of
    $(\set{\overline{\simpleCons_i}}, \Delta)$. 
  \end{itemize}
\end{definition}

A set of general unifiers can provide any unifier as a substitution
instance of one of its members.
\begin{definition}[Set of general unifiers]
  Let $C$ be a set of constraints and $\Delta$ a type environment. 

  A set of unifiers $M$ for $(C, \Delta)$
  is called \emph{set of general unifiers} if for any unifier $\omega$
  for $(C, \Delta)$ there is some unifier $\sigma \in M$ and a substitution
  $\lambda$ such that $\omega = \lambda   \circ \sigma$.
\end{definition}

A unification problem is \emph{finitary} if there is a finite set of
general unifiers for each constraint set $C$ and type environment $\Delta$.

\begin{theorem}[Soundness]
  \label{theo:unifySoundness}
   If $\unify{}(\consSet, \Delta) = {(\sigma,  \ol{Y} \triangleleft
     \ol{P})}$, then $\sigma$ is a unifier of $(\consSet,\Delta \cup \set{\ol{Y} <: \ol{P}})$. 
\end{theorem}

\begin{theorem}[Completeness]\label{theo:unifyCompleteness}
  $\unify{} (\consSet, \Delta)$ calculates {the set} of general
  unifiers for $(\consSet, \Delta)$. 
\end{theorem}

\begin{theorem}[Termination]\label{theo:unifyTermination}
  The \unify{} algorithm terminates on every finite input set.
\end{theorem}


\section{Soundness, Completeness and Complexity of Type Inference}
\label{sec:soundn-compl-type}

After showing that type unification is sound and complete, we
can now show that type inference \textbf{FJTypeInference} also is
sound and complete.
For the proofs of these theorems see appendix \ref{chapter:soundness-completenessProof}.

\begin{theorem}[Soundness] For all $\mv{\Pi}$, $\mv L$, $\mv{\Pi'}$, 
  $\fjtypeinference(\mv{\Pi},\texttt{L}) = \mv{\Pi}'$ implies $\mv{\Pi} \vdash \texttt{L} : \mv{\Pi}'$.
\end{theorem}
\begin{theorem}[Completeness]  For all $\mv{\Pi}$, $\mv L$, $\mv{\Pi'}$,
  $\mv{\Pi} \vdash \texttt{L} : \mv{\Pi}'$ implies there is a
    $\mv{\Pi}''$ with $\fjtypeinference(\mv{\Pi},\texttt{L}) = \mv{\Pi}''$,
  $\mv{\Pi} \vdash \texttt{L} : \mv{\Pi}''$, and the types of $\mv{\Pi}'$ are
  instances of $\mv{\Pi}''$.
\end{theorem}


\label{sec:complexity}

\begin{theorem}[NP-Hardness]
  \label{theo:np-hardness}
  The type inference algorithm for typeless Featherweight Java is NP-hard.
\end{theorem}


\begin{theorem}[NP-Completeness]
  \label{theo:np-completeness}
  The type inference algorithm for typeless Featherweight Java is NP-Complete.
\end{theorem}



\section{Related Work}
\label{sec:related-work}

\subsection{Formal models for Java}
\label{sec:formal-models-java}

There is a range of formal models for Java. Flatt et al
\cite{DBLP:conf/java/FlattKF99} define an elaborate model with
interfaces and classes and prove a type soundness result. They do not
address generics. Igarashi et al
\cite{DBLP:journals/toplas/IgarashiPW01} define Featherweight Java
and its generic sibling, Featherweight Generic Java. Their language is
a functional calculus reduced to the bare essentials, they develop the full metatheory, they
support generics, and study the type erasing transformation used by
the Java compiler. MJ \cite{UCAM-CL-TR-563} is a core calculus that
embraces imperative programming as it is targeted towards reasoning
about effects. It does not consider generics. Welterweight Java
\cite{DBLP:conf/tools/OstlundW10} and OOlong
\cite{DBLP:conf/sac/CastegrenW18} are different sketches for a core
language that includes concurrency, which none of the other core
languages considers. 

We chose to base our development on FGJ because it embraces a relevant
subset of Java without including too much complexity (e.g., no imperative
features, no interfaces, no concurrency). It seems that results for
FGJ are easily scalable to full Java. We leave the addition of these
feature to future work, as we see our results on FGJ as a first step
towards a formalized basis for global type inference for Java.

\subsection{Type inference}

Some object-oriented languages like Scala, C\#, and Java perform
\emph{local} type inference \cite{PT98,OZZ01}. Local type 
inference means that missing type annotations are recovered using only
information from adjacent nodes in the syntax tree without long distance
constraints. For instance, the type of a variable initialized with a
non-functional expression or the return type of a method can be
inferred. However, method argument types, in particular for recursive
methods, cannot be inferred by local type inference.

Milner's algorithm $\mathcal{W}$ \cite{DBLP:journals/jcss/Milner78} is
the gold standard for global type inference for languages with 
parametric polymorphism, which is used by ML-style languages. The fundamental idea
of the algorithm is to enforce type equality by many-sorted type
unification \cite{Rob65,MM82}. This approach is effective and results
in so-called principal types because many-sorted unification is
unitary, which means that there is at most one most general result.

Pl\"umicke \cite{Plue07_3} presents a first attempt to adopt Milner's
approach to Java. However, the presence of subtyping means that type
unification is no longer unitary, but still finitary. Thus, there is
no longer a single most general type, but any type is an instance of a
finite set of maximal types (for more details see Section
\ref{sec:unification}). Further work by the same author
\cite{plue15_2,plue17_2}, 
refines this approach by moving to a constraint-based algorithm and by
considering lambda expressions and Scale-like function types.
In Pluemicke's work there is no formal definition of its type system as a basis
of the type inference algorithm. Our contribution in this paper is a formal defined
type system. 

We rule out polymorphic recursion because its presence makes type
inference (but not type checking: see FGJ) undecidable. Henglein
\cite{DBLP:journals/toplas/Henglein93} as well as Kfoury et al
\cite{DBLP:journals/toplas/KfouryTU93} investigate type inference in
the presence of polymorphic recursion. They show that type inference
is reducible to semi-unification, which is undecidable
\cite{DBLP:journals/iandc/KfouryTU93}. However, the undecidability of
this problem apparently does not matter much in practice
\cite{DBLP:journals/tcs/EmmsL99}. 

Ancona, Damiani, Drossopoulou, Zucca \cite{ADDZ05} considered polymophic byte
code. For a type inference system this means that structural types have to be
inferred. As \textsf{Java} allows no structural types in \textsf{Java} this could be simulated by
generated interfaces. Pluemicke follows this approach in
\cite{plue16_1}. Furthermore Ancona et.al. considers only classes without
generics.

\subsection{Unification}
\label{sec:unification}

We reduce the type inference problem to constraint solving with
equality and subtype constraints.
The procedure presented in Section~\ref{sec:unify} is inspired by
polymorphic order-sorted unification which is used in logic 
programming languages with polymorphic order-sorted types
\cite{GS89,MH91,HiTo92,CB95}.

Smolka's thesis \cite{GS89} mentions the type unification problem
as an open problem. He gives  an incomplete type inference algorithm
for the logical language \textsf{TEL}. The reason for incompleteness
is the admission of subtype relationships between polymorphic types of
different arities as in  $\texttt{List(a)} \sub
\texttt{myLi(a,b)}$. The result is that the subtyping relation does
not fulfill the ascending chain condition.
For example, given  $\texttt{List(a)} \sub \texttt{myLi(a,b)}$, we obtain:
\begin{gather*}
  \texttt{List(a)} \sub \texttt{myLi(a,List(a))} \sub \texttt{myLi(a,myLi(a,List(a)))}  \sub \dots
\end{gather*}
However, this subtyping chain exploits covariant subtyping, which does
not apply to FGJ (but it would apply in the presence of wildcards).

Smolka's algorithm also fails sometimes in the absence of infinite
chains, although there is a unifier. 
For example, given $\texttt{nat} \sub \texttt{int}$ and the set
of subtyping constraints $\set{\mathtt{nat} \lessdot \mathtt{a},
  \mathtt{int} \lessdot \mathtt{a}}$, it returns the substitution
$\set{\mathtt{a} \mapsto \mathtt{nat}}$ generated from the first
constraint encountered. This substitution is not a solution
because $\set{\mathtt{int} \lessdot \mathtt{nat}}$ fails.
However, $\set{\mathtt{a}\mapsto \mathtt{int}}$ is a unifier, which
can be obtained by processing the constraints in a different order: from $\set{\mathtt{int} \lessdot \mathtt{a}, \mathtt{nat} \lessdot
  \mathtt{a}}$ the algorithm calculates the unifier 
$\set{\mathtt{a}\mapsto \mathtt{int}}$.

Hill and Topor  \cite{HiTo92} propose a polymorphically typed logic
programming language with subtyping. They restrict subtyping to type
constructors of the same arity,  which guarantees that all subtyping
chains are finite.
In this approach a \emph{most general type unifier (mgtu)} is
defined as an upper bound of different principal type unifiers. In
general, two type terms need not have an upper bound in the subtype ordering,
which means that there is no mgtu in the sense of Hill and Topor.
For example for  $\texttt{nat} \sub \texttt{int}$, $\texttt{neg} 
\sub \texttt{int}$, and the set of inequations $\set{\mathtt{nat} \lessdot
  \mathtt{a}$, $\mathtt{neg} \lessdot \mathtt{a}}$ the mgtu $\set{\mathtt{a} \mapsto \texttt{int}}$ is
determined. If the subtype ordering is extended by $\mathtt{int} \sub
\mathtt{index}$ and $\mathtt{int} \sub \mathtt{expr}$, then there are three
unifiers $\set{\mathtt{a} \mapsto \texttt{int}}$, $\set{\mathtt{a} \mapsto
  \mathtt{index}}$, and $\set{\mathtt{a} \mapsto 
\mathtt{expr}}$, but none of them is a mgtu \cite{HiTo92}.

The type system of \textsf{PROTOS-L} \cite{CB95} was
derived from \textsf{TEL} by disallowing any explicit subtype relationships
between polymorphic type constructors. 
Beierle \cite{CB95} gives a complete type unification algorithm, which can be extended to the
type system of Hill and Topor.
They also prove that the type unification problem is finitary.

Given the declarations  $\texttt{nat} \sub
\texttt{int}$, $\texttt{neg} \sub \texttt{int}$, $\mathtt{int} \sub
\mathtt{index}$, and $\mathtt{int} \sub \mathtt{expr}$, applying the
type unification algorithm of \textsf{PROTOS-L} to the set of
inequations $\set{\mathtt{nat} \lessdot
  \mathtt{a}$, $\mathtt{neg} \lessdot \mathtt{a}}$ yield three general
unifiers $\set{\mathtt{a} \mapsto \texttt{int}}$, $\set{\mathtt{a} \mapsto
  \mathtt{index}}$, and $\set{\mathtt{a} \mapsto \mathtt{expr}}$. 

Pl\"umicke \cite{plue09_1} realized that the type system of
\textsf{TEL} is related to subtyping in Java.
In contrast to \textsf{TEL}, where the ascending chain condition does
not hold,  Java with wildcards violates the descending chain
condition. For example, given $\exptypett{myLi}{b,a} \olsub
\exptypett{List}{a}$ we find:

\smallskip
{\centering
$\ldots \ \olsub\
\exptypett{myLi}{\exptypett{?\,$\extends$\,myLi}{\exptypett{\textrm{{\tt ?}}\,$\extends$\,List}{a},a},a}
\ \olsub\ \exptypett{myLi}{\exptypett{?\,$\extends$\,List}{a},a} \ \olsub\  \exptypett{List}{a}$\\}

\smallskip
Pl\"umicke \cite{plue09_1} solved the open problem of infinite chains
posed by Smolka \cite{GS89}.
He showed that in any infinite chain there is a finite number of elements such that
all other elements of the chain are instances of them. The resulting type
unification algorithm can be used for type inference of Java~5 with
wildcards \cite{Plue07_3}. As FGJ has no wildcards, we based our
algorithm on an earlier work \cite{Plue04_1}.
In contrast to that work, which only infers generic methods with
unbounded types, our algorithm  infers bounded generics.
To this end, we do not expand constraints
of the form $\tv{a} \lessdot \itype{N}$, where $\tv{a}$ is type variable and $\itype{N}$ is is a
non-variable type, but convert them to bounded type parameters of the form
\texttt{X extends N}. This change results in a significant reduction
of the number of solutions of the type
unification algorithm without restricting the generality of typings of
FGJ-programs. Unfortunately, constraints of the form $\itype{N} \lessdot \tv{a}$ have
to be expanded as FGJ (like Java) does not permit lower bounds for
generic parameters. If lower bounds were permitted  (as in Scala), the
number of solutions could be reduced even further.


\section{Conclusions}
\label{sec:conclusions}

This paper presents a global type inference algorithm applicable to
Featherweight Generic Java (FGJ). To this end, we define a language
\FGJGT that characterizes FGJ programs amenable to type inference: its
methods carry no type annotations and it does not permit polymorphic recursion.
This language corresponds to a strict subset of FGJ.
The inference algorithm is constraint based and is able to infer
method types with bounded generic types.

In future work, we plan to extend \FGJGT to a calculus with wildcards
inspired by Wild~FJ \cite{TEP05}. 
We also plan to extend the formal calculus with lambda expressions
(cf.\ \cite{BBDGV18}), but using true function types in place of
interface types.


\bibliography{peter,martin}

\begin{thebibliography}{10}

\bibitem{ADDZ05}
Davide Ancona, Ferruccio Damiani, Sophia Drossopoulou, and Elena Zucca.
\newblock Polymorphic bytecode: compositional compilation for java-like
  languages.
\newblock In Jens Palsberg and Mart{\'{\i}}n Abadi, editors, {\em Proceedings
  of the 32nd {ACM} {SIGPLAN-SIGACT} Symposium on Principles of Programming
  Languages, {POPL} 2005, Long Beach, California, USA, January 12-14, 2005},
  pages 26--37. {ACM}, 2005.
\newblock \href {https://doi.org/10.1145/1040305.1040308}
  {\path{doi:10.1145/1040305.1040308}}.

\bibitem{CB95}
Christoph Beierle.
\newblock Type inferencing for polymorphic order-sorted logic programs.
\newblock In {\em International Conference on Logic Programming}, pages
  765--779, 1995.

\bibitem{BBDGV18}
Lorenzo Bettini, Viviana Bono, Mariangiola Dezani-Ciancaglini, Paola Giannini,
  and Venneri Betti.
\newblock Java \& lambda: A featherweight story.
\newblock {\em Logical Methods in Computer Science}, 14(3:17):1--24, 2018.

\bibitem{UCAM-CL-TR-563}
G.M. Bierman, M.J. Parkinson, and A.M. Pitts.
\newblock {MJ: An imperative core calculus for Java and Java with effects}.
\newblock Technical Report UCAM-CL-TR-563, University of Cambridge, Computer
  Laboratory, April 2003.
\newblock URL: \url{https://www.cl.cam.ac.uk/techreports/UCAM-CL-TR-563.pdf},
  \href {https://doi.org/10.48456/tr-563} {\path{doi:10.48456/tr-563}}.

\bibitem{DBLP:conf/sac/CastegrenW18}
Elias Castegren and Tobias Wrigstad.
\newblock {OOlong}: an extensible concurrent object calculus.
\newblock In Hisham~M. Haddad, Roger~L. Wainwright, and Richard Chbeir,
  editors, {\em Proceedings of the 33rd Annual {ACM} Symposium on Applied
  Computing, {SAC} 2018, Pau, France, April 09-13, 2018}, pages 1022--1029.
  {ACM}, 2018.
\newblock \href {https://doi.org/10.1145/3167132.3167243}
  {\path{doi:10.1145/3167132.3167243}}.

\bibitem{DBLP:journals/tcs/EmmsL99}
Martin Emms and Hans Lei{\ss}.
\newblock Extending the type checker of standard {ML} by polymorphic recursion.
\newblock {\em Theor. Comput. Sci.}, 212(1-2):157--181, 1999.
\newblock \href {https://doi.org/10.1016/S0304-3975(98)00139-X}
  {\path{doi:10.1016/S0304-3975(98)00139-X}}.

\bibitem{DBLP:conf/java/FlattKF99}
Matthew Flatt, Shriram Krishnamurthi, and Matthias Felleisen.
\newblock A programmer's reduction semantics for classes and mixins.
\newblock In Jim Alves{-}Foss, editor, {\em Formal Syntax and Semantics of
  Java}, volume 1523 of {\em Lecture Notes in Computer Science}, pages
  241--269. Springer, 1999.
\newblock \href {https://doi.org/10.1007/3-540-48737-9_7}
  {\path{doi:10.1007/3-540-48737-9_7}}.

\bibitem{MH91}
Michael Hanus.
\newblock Parametric order-sorted types in logic programming.
\newblock {\em Proc. TAPSOFT 1991}, LNCS(394):181--200, 1991.

\bibitem{DBLP:journals/toplas/Henglein93}
Fritz Henglein.
\newblock Type inference with polymorphic recursion.
\newblock {\em {ACM} Trans. Program. Lang. Syst.}, 15(2):253--289, 1993.
\newblock \href {https://doi.org/10.1145/169701.169692}
  {\path{doi:10.1145/169701.169692}}.

\bibitem{HiTo92}
Patricia~M. Hill and Rodney~W. Topor.
\newblock A {S}emantics for {T}yped {L}ogic {P}rograms.
\newblock In Frank Pfenning, editor, {\em Types in Logic Programming}, pages
  1--62. MIT Press, 1992.

\bibitem{DBLP:journals/toplas/IgarashiPW01}
Atsushi Igarashi, Benjamin~C. Pierce, and Philip Wadler.
\newblock Featherweight {Java}: A minimal core calculus for {Java} and {GJ}.
\newblock {\em {ACM} Trans. Program. Lang. Syst.}, 23(3):396--450, 2001.
\newblock \href {https://doi.org/10.1145/503502.503505}
  {\path{doi:10.1145/503502.503505}}.

\bibitem{DBLP:journals/toplas/KfouryTU93}
A.~J. Kfoury, Jerzy Tiuryn, and Pawel Urzyczyn.
\newblock Type reconstruction in the presence of polymorphic recursion.
\newblock {\em {ACM} Trans. Program. Lang. Syst.}, 15(2):290--311, 1993.
\newblock \href {https://doi.org/10.1145/169701.169687}
  {\path{doi:10.1145/169701.169687}}.

\bibitem{DBLP:journals/iandc/KfouryTU93}
A.~J. Kfoury, Jerzy Tiuryn, and Pawel Urzyczyn.
\newblock The undecidability of the semi-unification problem.
\newblock {\em Inf. Comput.}, 102(1):83--101, 1993.
\newblock \href {https://doi.org/10.1006/inco.1993.1003}
  {\path{doi:10.1006/inco.1993.1003}}.

\bibitem{MM82}
A.~Martelli and U.~Montanari.
\newblock An efficient unification algorithm.
\newblock {\em ACM Transactions on Programming Languages and Systems},
  4:258--282, 1982.

\bibitem{DBLP:journals/jcss/Milner78}
Robin Milner.
\newblock A theory of type polymorphism in programming.
\newblock {\em J. Comput. Syst. Sci.}, 17(3):348--375, 1978.
\newblock \href {https://doi.org/10.1016/0022-0000(78)90014-4}
  {\path{doi:10.1016/0022-0000(78)90014-4}}.

\bibitem{OZZ01}
Martin Odersky, Matthias Zenger, and Christoph Zenger.
\newblock Colored local type inference.
\newblock {\em Proc. 28th ACM Symposium on Principles of Programming
  Languages}, 36(3):41--53, 2001.

\bibitem{DBLP:conf/tools/OstlundW10}
Johan {\"{O}}stlund and Tobias Wrigstad.
\newblock Welterweight java.
\newblock In Jan Vitek, editor, {\em Objects, Models, Components, Patterns,
  48th International Conference, {TOOLS} 2010, M{\'{a}}laga, Spain, June 28 -
  July 2, 2010. Proceedings}, volume 6141 of {\em Lecture Notes in Computer
  Science}, pages 97--116. Springer, 2010.
\newblock \href {https://doi.org/10.1007/978-3-642-13953-6_6}
  {\path{doi:10.1007/978-3-642-13953-6_6}}.

\bibitem{PT98}
Benjamin~C. Pierce and David~N. Turner.
\newblock Local type inference.
\newblock In {\em Proceedings of the 25th ACM SIGPLAN-SIGACT symposium on
  Principles of programming languages}, POPL '98, pages 252--265, 1998.

\bibitem{Plue04_1}
Martin Pl{\"u}micke.
\newblock Type {U}nification in \textsf{Generic--Java}.
\newblock In Michael Kohlhase, editor, {\em Proceedings of 18th {I}nternational
  {W}orkshop on {U}nification ({U}{N}{I}{F}'04)}, Cork, July 2004.

\bibitem{Plue07_3}
Martin Pl{\"u}micke.
\newblock Typeless {P}rogramming in \textsf{{J}ava 5.0} with {W}ildcards.
\newblock In Vasco Amaral, Lu\'is Veiga, Lu\'is Marcelino, and H.~Conrad
  Cunningham, editors, {\em 5th {I}nternational {C}onference on {P}rinciples
  and {P}ractices of {P}rogramming in {J}ava}, volume 272 of {\em ACM
  International Conference Proceeding Series}, pages 73--82, September 2007.

\bibitem{plue09_1}
Martin Pl{\"u}micke.
\newblock Java type unification with wildcards.
\newblock In Dietmar Seipel, Michael Hanus, and Armin Wolf, editors, {\em 17th
  International Conference, INAP 2007, and 21st Workshop on Logic Programming,
  WLP 2007, W\"urzburg, Germany, October 4-6, 2007, Revised Selected Papers},
  volume 5437 of {\em Lecture Notes in Artificial Intelligence}, pages
  223--240. Springer-Verlag Heidelberg, 2009.

\bibitem{plue15_2}
Martin Pl{\"{u}}micke.
\newblock More type inference in {J}ava 8.
\newblock In Andrei Voronkov and Irina Virbitskaite, editors, {\em Perspectives
  of System Informatics - 9th International Ershov Informatics Conference,
  {PSI} 2014, St. Petersburg, Russia, June 24-27, 2014. Revised Selected
  Papers}, volume 8974 of {\em Lecture Notes in Computer Science}, pages
  248--256. Springer, 2015.

\bibitem{plue16_1}
Martin Pl{\"{u}}micke.
\newblock Structural type inference in java-like languages.
\newblock In {\em Gemeinsamer Tagungsband der Workshops der Tagung Software
  Engineering 2016 {(SE} 2016), Wien, 23.-26. Februar 2016.}, pages 109--113,
  2016.
\newblock URL: \url{http://ceur-ws.org/Vol-1559/paper09.pdf}.

\bibitem{plue17_2}
Martin Pl\"{u}micke and Andreas Stadelmeier.
\newblock Introducing {S}cala-like function types into {J}ava-{TX}.
\newblock In {\em Proceedings of the 14th International Conference on Managed
  Languages and Runtimes}, ManLang 2017, pages 23--34, New York, NY, USA, 2017.
  ACM.
\newblock \href {https://doi.org/10.1145/3132190.3132203}
  {\path{doi:10.1145/3132190.3132203}}.

\bibitem{Rob65}
J.~A. Robinson.
\newblock A machine-oriented logic based on the resolution principle.
\newblock {\em Journal of ACM}, 12(1):23--41, January 1965.

\bibitem{DBLP:conf/aplas/Simonet03}
Vincent Simonet.
\newblock Type inference with structural subtyping: {A} faithful formalization
  of an efficient constraint solver.
\newblock In Atsushi Ohori, editor, {\em Programming Languages and Systems,
  First Asian Symposium, {APLAS} 2003}, volume 2895 of {\em Lecture Notes in
  Computer Science}, pages 283--302, Beijing, China, November 2003. Springer.
\newblock \href {https://doi.org/10.1007/978-3-540-40018-9_19}
  {\path{doi:10.1007/978-3-540-40018-9_19}}.

\bibitem{GS89}
Gert Smolka.
\newblock {\em Logic Programming over Polymorphically Order-Sorted Types}.
\newblock PhD thesis, Department Informatik, University of Kaisers\-lautern,
  Kaiserslautern, Germany, May 1989.

\bibitem{TEP05}
Mads Torgersen, Erik Ernst, and Christian~Plesner Hansen.
\newblock Wild {F}{J}.
\newblock In Philip Wadler, editor, {\em Proceedings of FOOL 12}, Long Beach,
  California, USA, January 2005. ACM, School of Informatics, University of
  Edinburgh.
\newblock URL: \url{http://homepages.inf.ed.ac.uk/wadler/fool/}.

\end{thebibliography}

\appendix

\section{Unify Soundness Proof}\label{chapter:unifySoundnessProof}
\begin{proof}
    We show theorem \ref{theo:unifySoundness} by going backwards over every step of the algorithm.
    Let $\sigma = \set {a_1 \mapsto T_1, \ldots , a_n \mapsto T_n}$ and $\set{\ol{Y} <: \ol{P}}$ be the result of a $\unify{}(\consSet, \Delta)$ call.
    We show for every constraint in the input set $(a \lessdot b) \in \consSet_{in}$ and $(c \doteq d) \in \consSet_{in}$:
    $\Delta, \ol{Y} <: \ol{P} \vdash \sigma(a) \olsub \sigma(b)$ and $\sigma(c) = \sigma(d)$\\
    
    We now consider each step of the \unify{} algorithm
    which transforms the input set of constraints $C$ to a set $C'$
    If $\sigma$ is an unifier of $C'$, then $\sigma$ is an unifier of $C$, too.
    
    \begin{description}
    \item[Step 6] The last step does not change the constraint set.
    \item[Step 5]
    A unifier which is correct for $a \doteq b$ is also correct for $a \lessdot b$.
    The transformation $C' = [a/b]C$ does not change this.
    
    \item[Step 4]
    The constraint sets are not altered here.
    
    \item[Step 3]
    An unifier $\sigma$ that is correct for a constraint set
    $[\itype{T}/a]C \cup \set{a \doteq \itype{T}}$ is also correct for
    the set $C \cup \set{a \doteq \itype{T}}$.
    From the constraint $(a \doteq \itype{T})$ it follows that $\sigma(a) = \itype{T}$.
    This means that $\sigma(C) = \sigma([\itype{T}/a]C)$,
    because every occurence of $a$ in $C$ will be replaced by $\itype{T}$ anyways when using the unifier $\sigma$.
    
    \item[Step 2]
    This step transforms constraints of the form $\exptype{C}{\ol{X}} \lessdot a$ and
    $\set{a \lessdot \exptype{C}{\ol{X}}, a \lessdot^* b}$ jinto sets of or-constraints.
    We can show that if there is a resulting set of constraints which has $\sigma$ as its correct unifier
    then $\sigma$ also has to be a correct unifier for the constraints before this transformation.
    
    We look at each transformation done in step 2:
    \begin{description}
    \item[$\set{\exptype{C}{\ol{T}} \lessdot a} \in C \to \set{a \doteq [\ol{T}/\ol{X}]N} \in C'$:]
    If $\exptype{C}{\ol{X}} <: N$ and $\sigma$ is correct for $(a \doteq [\ol{T}/\ol{X}]N)$
    then $\sigma$ is also correct for $(\exptype{C}{\ol{T}} \lessdot a))$.
    When substituting $a$ for $[\ol{T}/\ol{X}]N$ we get 
    $(\exptype{C}{\ol{T}} \lessdot [\ol{T}/\ol{X}]N)$
    , which is correct because $\exptype{C}{\ol{X}} <: \exptype{C}{\ol{Y}}$
    (see \texttt{S-CLASS} rule).
    \item[$\set{a \lessdot \exptype{C}{\overline{T}},\ a \lessdot^* b} \in C \to \set{T \lessdot b, a \lessdot \exptype{C}{\overline{T}},\ a \lessdot^* b} \in C'$]
    obviously.
    \item[$\set{a \lessdot \exptype{C}{\overline{T}},\ a \lessdot^* b} \in C \to \set{a \doteq [\ol{T}/\ol{X}]N, a \lessdot \exptype{C}{\overline{T}},\ a \lessdot^* b} \in C'$]
    This is the same as in the first transformation.
    Here we can also show correctness via the \texttt{S-CLASS} rule.
    
    \end{description}
    
    \item[Step 1]
    \begin{description}
    \item[erase-rules] remove correct constraints from the constraint set.
    A unifier $\sigma$ that is correct for the constraint set $C$
    is also correct for $C \cup \set{\theta \doteq \theta}$
    and $C \cup \set{\theta \lessdot \theta'}$, when $\theta \leq \theta'$.
    \item[swap-rule] does not change the unifier for the constraint set.
    $\doteq$ is a symmetric operator and parameters can be swapped freely.
    \item[match] The subtype relation is transitive, so if there is a correct solution for
    $a \lessdot \exptype{C}{\ol{X}}, \exptype{C}{\ol{X}} \lessdot \exptype{D}{\ol{Y}}$
    then this solution would also apply for $a \lessdot \exptype{C}{\ol{X}} \lessdot \exptype{D}{\ol{Y}}$
    or $a \lessdot \exptype{D}{\ol{Y}}$.
    \item[adopt] An unifier which is correct for $C \cup \set{a \lessdot \exptype{C}{\ol{X}}, b \lessdot^* a, b \lessdot \exptype{D}{\ol{Y}}, b \lessdot \exptype{C}{\ol{X}}}$
    is also correct for $C \cup \set{a \lessdot \exptype{C}{\ol{X}}, b \lessdot^* a, b \lessdot \exptype{D}{\ol{Y}}}$.
    \item[adapt] If there is a $\sigma$ which is a correct unifier for a set
    $C \cup \set{ \exptype{C}{[\ol{A}/\ol{X}]\ol{Y}} \doteq \exptype{C}{\ol{B}}}$ then it is also
    a correct unifier for the set $C \cup \set{ \exptype{D}{\ol{A}} \lessdot \exptype{C}{\ol{B}}}$,
    if there is a subtype relation $\exptype{D}{\ol{X}} \leq^* \exptype{C}{\ol{Y}}$.
    To make the set $C \cup \set{ [\ol{A}/\ol{X}]\exptype{C}{\ol{Y}} \doteq \exptype{C}{\ol{B}}}$ the unifier 
    $\sigma$ must satisfy the condition $\sigma([\ol{A}/\ol{X}]\ol{Y}) = \sigma(\ol{B})$.
    By substitution we get $C \cup \set{ \exptype{D}{\ol{A}} \lessdot \exptype{C}{[\ol{A}/\ol{X}]\ol{Y}}}$
    which is correct under the \texttt{S-CLASS} rule.
    \item[reduce] The \texttt{reduce} rule is obviously correct under the FJ typing rules.
    \end{description}
    
    \item[OrConstraints]
    If $\sigma$ is a correct unifier for one of the constraint sets in $C_{set}$
    then it is also a correct unifier for the input set $\consSet_{in}$.
    When building the cartesian product of the \textbf{OrConstraints} every possible
    combination for $\consSet_{in}$ is build.
    No constraint is altered, deleted or modified during this step.
    \end{description}
    \end{proof}
    
    \hfill $\square$

\section{Unify Completeness Proof}\label{chapter:unifyCompletenessProof}
\begin{proof}

We proof theorem \ref{theo:unifyCompleteness} by assuming there exists a general unifier $\sigma = \set {a_1 \mapsto T_1, \ldots , a_n \mapsto T_n}$.
We then look at every step of the algorithm, which alters the set of constraints $C$.
We show that if $\sigma$ is general unifer for the input then $\sigma$ is a
general unifer for the altered set of constraints. This means that no solution
is excluded.


\begin{description}
\item[Step 1:]
The first step applies the seven rules from figure
\ref{fig:fgjreduce-rules} and \ref{fig:fgjerase-rules}.

\textbf{erase-rule:} The constraint $\tv{a} \doteq \tv{a}$ is true for every unifier and can be removed.

\textbf{swap-rule:} $\doteq$ is a symmetric operator and parameters can be swapped freely.
This operation does not change the meaning of the constraint set.

\textbf{match-rule:}
If there is a solution for $a \lessdot \exptype{C}{\ol{\itype{T}}}, a \lessdot \exptype{D}{\ol{\itype{U}}}$,
this is also a solution for $a \lessdot \exptype{C}{\ol{\itype{T}}}, \exptype{C}{\ol{\itype{T}}} \lessdot \exptype{D}{\ol{\itype{U}}}$.
A correct unifier $\sigma$ has to find a type for $a$, which complies with $a \lessdot \exptype{C}{\ol{\itype{T}}}$ and $a \lessdot \exptype{D}{\ol{\itype{U}}}$.
Due to the subtyping relation being transitive this means that $\sigma(a) \lessdot \exptype{C}{\ol{\itype{T}}} \lessdot \exptype{D}{\ol{\itype{U}}}$.

\textbf{adopt-rule:} Subtyping in FJ is transitive,
which allows us to apply the adopt rule without excluding any possible unifier.

\textbf{adapt-rule:} Every solution which is correct for the constraints
$Eq \cup \set{ \exptype{C}{[\ol{\itype{A}}/\ol{X}]\ol{\itype{T}}} \doteq \exptype{C}{\ol{\itype{U}}}}$ is also
a correct solution for the set $Eq \cup \set{ \exptype{D}{\ol{\itype{A}}} \lessdot \exptype{C}{\ol{\itype{U}}}}$.
According to the FGJ \texttt{S-CLASS} rule there can only be a possible solution for 
$\exptype{C}{[\ol{\itype{A}}/\ol{X}]\ol{\itype{T}}} \doteq \exptype{C}{\ol{\itype{U}}}$
if $\ol{\itype{U}} = [\ol{\itype{A}}/\ol{X}]\ol{\itype{T}}$.
Therefore this transformation does not remove any possible solution from the constraint set.

\textbf{reduce-rule:}

If $\sigma$ is a unifier of $\exptype{D}{\ol{\itype{T}}} \lessdot
\exptype{D}{\ol{\itype{U}}}$
then $\sigma$ is a unifier of $\ol{\itype{T}} = \ol{\itype{U}}$.
Therefore this step does not remove a possible solution.

\textbf{equals-rule:}
This rule removes a circle in the constraints.
This does not remove a solution.

\item[Step 2:]
The second step of the algorithm eliminates $\lessdot$-constraints
by replacing them with $\doteq$-constraints.
For each $(\exptype{C}{\ol{\itype{T}}} \lessdot b), (b \lessdot \exptype{C}{\ol{\itype{U}}})$ constraint the algorithm builds a set with every
possible supertype of $\exptype{C}{\ol{\itype{T}}}$.
So if there is a correct unifier $\sigma$ for the constraints before this conversion there will be at least one set of
constraints for which $\sigma$ is a correct unifier.

Additionally this step resolves constraints of the form $\tv{a} \lessdot \exptype{C}{\ol{\itype{T}}}, \tv{a} \lessdot^* \tv{b}$.
We generate an or-constraint with every possible combination for $\tv{b}$.
This includes every possible solution for $\tv{b}$ and therefore does not remove a possible solution.
This is due to the fact that \TFGJ does not allow lower bounds for generic variables.

\item[Step 3:]
In the third step the \textbf{substitution}-rule is applied.
If there is a constraint $\tv{a} \doteq \type{N}$ then there is no other way to fulfill the constraint set
than replacing $\tv{a}$ with $\type{N}$.
This does not remove a possible solution.

\item[Step 4:]
None of the constraints get modified.

\item[Step 5:]
If the algorithm advances to this step we further only work on constraint sets in solved form.
This means there are only four kinds of constraints left:
($a \doteq \itype{T}$), ($a \lessdot \itype{T}$), ($a \doteq b$) and ($a \lessdot b$) with $a$ and $b$ as type variables.


The FGJ language does not allow subtype constraints for generic types.
A constraint like $(a \lessdot b)$ in a solution could be inserted as the typing shown in the example below.
But this is not allowed by the syntax of FGJ.
That is why we can treat this constraint as $(a \doteq b)$.


\textit{Example:}
This would be a valid Java program but is not allowed in FGJ:
\begin{lstlisting}
class Example {
    <A extends Object, B extends A> A id(B a){
    return a;
    }
}
\end{lstlisting}

By replacing all ($a \lessdot b$) constraints with ($a \doteq b$) we
    do not remove the general unifier $\sigma$ as $a$ and $b$ are not substituted
    in $\sigma$. 

\item[Step 6:]
In the last step all the constraint sets, which are in solved form, are converted to unifiers.

We see that only a constraint set which has no unifier does not reach solved form.
We showed that in none of the steps of the \unify{} algorithm we exclude a possible unifier.
Also we showed that after we reach step 5 only constraint sets with a correct unifier are in solved form.
By removing all constraint sets which are not in solved form the algorithm does not
remove a possible correct unifier.

If we assume that there is a possible general unifier $\sigma$ for the input set $\consSet_{in}$
and the \unify{} algorithm does not exclude any of the possible unifiers,
then the result \unify{} contains the general unifier.
\hfill $\square$
\end{description}
\end{proof}

\section{Unify Termination Proof}\label{chapter:unifyTerminationProof}

The \unify{} algorithm gets called with a set of input constraints.
After resolving the \textbf{OrConstraints} we end up with multiple $Eq$ sets.
Afterwards the algorithm iterates over each of those sets (see Chapter \ref{sec:unify}).
We will show that \unify{} terminates on each of those sets by showing,
that each step of the algorithm removes at least one type variable
until the finishing state is reached.
The finishing state for a constraint set is reached when step 3 is not able to substitute a type variable.
This is checked by step 4 of the algorithm.
Then the $Eq$ set is either in solved form or determined to be unsolvable.

\textit{Proof:}
The \unify{} algorithm reduces the amount of type variables with every iteration.
No step adds a new type variable to the constraint set.
Additionally we have to show that the first step of the algorithm also terminates on every finite input set.

\begin{description}
\item[Step 1] 
Step 1 of the algorithm always terminates. \textit{Proof:}
Every rule either removes a $\lessdot$ constraint or reduces a $\exptype{C}{\ol{X}}$ to $\ol{X}$ inside a constraint.
None of the rules add a new $\lessdot$ constraint or a $\exptype{C}{\ol{X}}$ type to the constraint set.
Step 1 has to come to a stop once there are no more $\lessdot$ constraints or $\exptype{C}{\ol{X}}$ types to reduce.

The rule \textbf{match} seems to generate a new $\exptype{C}{\ol{X}}$ constraint,
but the $\exptype{C}{\ol{X}} \lessdot \exptype{D}{\ol{Y}}$ constraint added by \texttt{match}
will be changed immidiatly into a $\doteq$ constraint by the \texttt{adapt} rule.
Afterwards the \texttt{reduce1} rule will remove this freshly added $\exptype{C}{\ol{X}}$ type.
So effectively a $\doteq$ constraint is removed by this rule in combination with \texttt{adapt} and \texttt{reduce1}.

The \textbf{adopt} rule seems to generate a new $\lessdot$ constraint.
But the \texttt{adopt} rule triggers two other rules. The \texttt{match} and the \texttt{adapt} rule.
\begin{enumerate}
  \item We start with the \texttt{adopt} rule: \\
   $
  \begin{array}[c]{ll}
      \begin{array}[c]{l}
          Eq \cup \, \set{a \lessdot
          \exptype{C}{\ol{X}},
          b \lessdot^* a, b \lessdot \exptype{D}{\ol{Y}}} \\ 
          \hline
          \vspace*{-0.4cm}\\
          Eq \cup \set{
          a \lessdot
          \exptype{C}{\ol{X}},
          b \lessdot^*
          a
          , b \lessdot \exptype{C}{\ol{X}}
          , b \lessdot \exptype{D}{\ol{Y}}
          }
      \end{array}
      \end{array}
      $
  \item We can now apply the \texttt{match} rule to the two resulting $(b \lessdot \ldots)$-constraints.
  If this is not possible due to type \texttt{C} not being a subtype of \texttt{D} or vice versa,
  then the $Eq$ set has no possible solution and \unify{} would terminate as fail $Uni = \emptyset$: \\
  $
  \begin{array}[c]{ll}
  \begin{array}[c]{l}
      Eq \cup \, \set{b \lessdot
      \exptype{C}{\ol{X}},
      b \lessdot
      \exptype{D}{\ol{Y}}} \\ 
      \hline
      \vspace*{-0.4cm}\\
      Eq \cup \set{b \lessdot \exptype{C}{\ol{X}}
      , \exptype{C}{\ol{X}} \lessdot \exptype{D}{\ol{Y}}}
  \end{array}
  & \exptype{C}{\ol{Z}} <: \exptype{D}{\ol{N}} 
  \end{array}
      $\\
  \item The constraint added by the \texttt{match} rule fits the \texttt{adapt} rule, which we apply in the next step: $
  \begin{array}[c]{ll}
  \begin{array}[c]{l}
     Eq \cup \, \set{\exptype{C}{\ol{X}} \lessdot
      \exptype{D}{\ol{Y}}} \\ 
    \hline
    \vspace*{-0.4cm}\\
    Eq \cup \set{\exptype{D}{[ \ol{X} / \ol{Z} ]\ol{N}}
    \doteq \exptype{D}{\ol{Y}}}
  \end{array}
  & \exptype{C}{\ol{Z}} <:\ \exptype{D}{\ol{N}}
  \end{array}
  $
  \end{enumerate}
In the end we have the conversion:\\
\begin{align*}\ddfrac{
  Eq \cup \, \set{a \lessdot
  \exptype{C}{\ol{X}},
  b \lessdot^* a, b \lessdot \exptype{D}{\ol{Y}}}
}{
  Eq \cup \set{a \lessdot
  \exptype{C}{\ol{X}},
  b \lessdot^* a, b \lessdot \exptype{D}{\ol{Y}}, \exptype{D}{[ \ol{X} / \ol{Z} ]\ol{N}}
  \doteq \exptype{D}{\ol{Y}}}
}\end{align*}

We can see now, that only a $\doteq$ constraint is added.
The \texttt{adopt} alone adds a $\lessdot$ constraint,
but due to the fact that it is always used together with \texttt{match} and \texttt{adapt} it effectively just adds a $\doteq$ constraint.

\item[Step 2] This step does not add new type variables to the constraint set.
\item[Step 3] The third step of the \unify{} algorithm removes at least one type variable
from the constraint set or otherwise does not alter $Eq$ at all.
If $Eq$ is not altered the algorithm terminates in the next step.
The type variable is not completely removed but stays inside $Eq$ only in one $a \doteq N$ constraint.
All other occurences are replaced by $N$.
The \texttt{subst} step can therefore only be executed once per type variable.
\end{description}

We see that with each iteration over the steps 1-3 at least one type variable is removed from the constraint set.
Due to the fact that there is never added a fresh type variable during the \unify{} algorithm,
the algorithm will terminate for any given finite set of constraints. \hfill $\square$

\section{Soundness, Completeness and Complexity Proofs}\label{chapter:soundness-completenessProof}
We show soundness and completeness by a case analysis over the type rules given by \TFGJ{}.
We will show that the constraints generated by \fjtype{} mirror
the \TFGJ{} type rules.
The \unify{} algorithm fullfils the constraints generated by \fjtype{} and
the \unify{} algorithm is sound and complete (see theorem \ref{theo:unifySoundness} and \ref{theo:unifyCompleteness}).
Therefore the $\fjtypeinference$ algorithm is sound and
complete, if the constraints generated by \fjtype{} mirror
the \TFGJ{} type rules.

As the \unify{} algorithm determines the set of general unifers it
  holds true that the types of $\mv{\Pi}'$ are
  instances of $\mv{\Pi}''$.

Now we show that the constraints generated by {\fjtype{}}
represent the type rules given in chapter \ref{chapter:type-rules}.




The constraint generation starts with generating the method assumptions for the current class $\overline{\methodAssumption}$.
The constraints in $\consSet_o$ ensure valid overriding.
The constraints in $\consSet_m$ ensure that every type parameter has a bound.

We compare the constraints generated by the \typeExpr{} function with the appropriate type rule from \TFGJ:
\begin{description}
  \item [Local var]
  No constraints are generated.
  \item[Method invocation]
By direct comparison we show that each of the generated constraints applies the same restrictions than the \texttt{GT-INVK} rule.
The \texttt{GT-INVK} rule states the condition $\textit{mtype}(m, \textit{bound}_\triangle(T_0), \mtypeEnvironment) = \exptype{}{\ol{Y} \triangleleft \ol{P}}\ \ol{U} \to U$.
The constraint $\type{R} \lessdot \exptype{C}{\ol{\tv{a}}}$ assures that the type of the expression $e_0$ contains the method \texttt{m}.
The type variables $\ol{\tv{b}}$ represent instances of $\ol{Y}$.
When calling a local method the generic variables $\ol{Y} \triangleleft \ol{P}$ are not present and $\ol{\tv{b}}$ is empty.
This is correct, because the \TFGJ type rules prevent local method calls to be polymorphic.

We generate constraints according to the \texttt{GT-INVK} rule:\\
\begin{small}
\begin{tabularx}{\linewidth}{lX|Xl}
  \textbf{\TFGJ{} Type rule} &&& \textbf{Constraints} \\
  $\mathtt{\environmentvdash e_0 : T_0 }$ &&&
    $({R}, \consSet_R) = \typeExpr(({\mtypeEnvironment} ;\overline{\localVarAssumption}), \texttt{e})$\\ 
  $\mathtt{\mathit{bound}_\Delta (T_0)}$ &&& $\type{R} \lessdot \exptype{C}{\ol{\tv{a}}}$ \\
  $\mathtt{\mathit{mtype}(m, \mathit{bound}_\Delta (T_0), \Pi)}$ &&& Lookup in the assumptions \\
 $\environmentvdash \ol{e} : \ol{S}$ &&& $\forall \texttt{e}_i \in \ol{e} : (R_i, \consSet_i) = \typeExpr(({\mtypeEnvironment} ;
 \overline{\localVarAssumption}), \texttt{e}_i)$\\
 $\Delta \vdash \ol{S} \subeq  [\ol{V}/\ol{Y}]\ol{U}$ &&& $ \overline{R} \lessdot [\ol{\tv{b}}/\ol{Y}][\ol{\tv{a}}/\ol{X}]\overline{T}$\\
 $\environmentvdash \mathtt{e_0.m(\overline{e}) : [\ol{V}/\ol{Y}]U }$ &&& $ \tv{a} \doteq [\ol{\tv{b}}/\ol{Y}][\ol{\tv{a}}/\ol{X}]{T}$ \\
 $\Delta \vdash \ol{V} \subeq  [\ol{V}/\ol{Y}]\ol{P}$ &&& $\ol{\tv{b}} \lessdot [\ol{\tv{b}}/\ol{Y}][\ol{\tv{a}}/\ol{X}]\ol{P}$ \\
\end{tabularx}
\end{small}

No constraint is needed to ensure $\ol{V}\ \texttt{ok}$.

\item[Field access]
The constraint generation behaves mostly the same as method invocation.
We also generate or-constraints in the case of multiple classes containing a field with the same name.
The field types are already given in the input and need not to be inferred.

 \item[Constructor] We generate constraints according to the \texttt{GT-NEW} rule:\\
 \begin{tabularx}{\linewidth}{lX|Xl}
  \textbf{\TFGJ{} Type rule} &&& \textbf{Constraints} \\
  $\environmentvdash \ol{e} : \ol{S}$ &&& $\forall \texttt{e}_i \in \overline{\texttt{e}} : (R_i, \consSet_i) = \typeExpr(({\mtypeEnvironment} ;
  \overline{\localVarAssumption}), \texttt{e}_i)$\\
  $\environmentvdash \ol{S} <: \ol{T}$ &&& $\overline{R} \lessdot \overline{T}$\\
  $\type{N}\ \texttt{ok}$ &&& $\overline{\tv{a}} \lessdot
  [\overline{\tv{a}}/\ol{X}]\ol{N}$\\
  $\mathtt{\mathit{bound}_\Delta (T_0)}$ &&& ${R} \lessdot \exptype{C}{\ol{\tv{a}}}$ \\
  $\mathtt{\mathit{fields}(\mathit{bound}_\Delta (T_0))}$ &&& $\tv{a} \doteq [\overline{\tv{a}}/\ol{X}]\type{T}$
\end{tabularx}

\item[Cast]
A cast can either be an upcast or a downcast or a so called stupid cast.
See respective type rules \texttt{GT-UCAST}, \texttt{GT-DCAST}, \texttt{GT-SCAST} in chapter \ref{chapter:type-rules}.
We assume that each given type in our input set is well-formed.
Therefore the cast type $\mathtt{N}$ is well formed too.
So every possible type of cast is allowed in \TFGJ, therefore no restrictions in form of constraints are needed.

\end{description}

\hfill $\square$

\section{NP-Hard Complexity Proof}\label{section:np-hard}
This section will show this by reducing the boolean satisfiability problem (SAT) to the \fjtypeinference{} algorithm.

\begin{figure}
\begin{lstlisting}
  class True extends Object{
  }
  class False extends Object{
  }

  class Nand1 extends Object{
    False nand(True a, True b){ return new False(); }
  }
  class Nand2 extends Object{
    True nand(False a, True b){ return new True(); }
  }
  class Nand3 extends Object{
    True nand(True a, False b){ return new True(); }
  }
  class Nand4 extends Object{
    True nand(False a, False b){ return new True(); }
  }

  class SATExample extends Object{
    True f;

    sat(v1, v2, v3, o1, o2){
      return o1.nand(v1, o2.nand(v2, v3));
    }

    forceSATtoTrue(v1, v2, v3, o1, o2){
      return new SATExample(this.sat(v1, v2, v3, o1, o2));
    }
  }
\end{lstlisting}

\caption{Representation for a SAT problem in FJ code}
\label{fig:fjSATcode}
\end{figure}

Any given boolean expression $B$ can be transformed to a typeless FJ program.
A type inference algorithm finding a possible typisation of this FJ program also solves the boolean expression $B$.
Figure \ref{fig:fjSATcode} shows an example of this.
The classes \texttt{True}, \texttt{False} and \texttt{Operations} always stay the same.
Here we assume that the boolean expression only consists out of $\neg \land$ (NAND) operators.
Now any boolean expression $B_\text{in} = v_1 \land \neg (v_2 \land v_3) \land \ldots$ can be expressed as a Java method.
The example in figure \ref{fig:fjSATcode} represents the problem $B_\text{in} = \neg(v_1 \land \neg (v_2 \land v_3))$.
Additionally we force the return type of the \texttt{sat} method to have the type \texttt{True}
by instancing the \texttt{SATExample} class, which requires the type \texttt{True}.
When using the \fjtypeinference{} algorithm on the generated FJ code it will
assign each parameter of the \texttt{sat} method with either the type \texttt{True} or \texttt{False}.
This represents a valid assignment for the expression $B_\text{in}$.
If \fjtypeinference{} fails to compute a solution the $B_\text{in}$ has no possible solution.
A correct solution for the \texttt{sat} method in figure \ref{fig:fjSATcode} would be:\\
\texttt{True sat(False v1, True v2, True v3, Nand4 o1, Nand1 o2)}

Any SAT problem can be transferred in polynomial time to a typeless FJ program.
Every literal $v$ in the SAT problem becomes a method parameter of the \texttt{sat} method, as well as every instance of a NAND operator used.

This reduction of SAT to our type inference algorithm proofs that its
complexity is at least NP-Hard.
\hfill $\square$

\section{NP-Complete Complexity Proof}\label{section:np-complete}
We know the algorithm is NP-hard (see \ref{theo:np-completeness}).
To proof NP-Completeness we have to show that it is possible to verify a solution in polynomial time.
The verification of a type solution is the FJ typecheck.

It is easy to see that the expression typing rules can be checked in polynomial time as long as subtyping between two types is verifiable in polynomial time.

Subtyping is also solvable in polynomial time in \TFGJ{}.
Assume $\exptype{C}{\ol{X}} \leq \exptype{D}{\ol{Y}}$ with the number of generics $\ol{X}$ and $\ol{Y}$ less or equal $n$.
Also the number of classes in the subtyperelation is less or equal to $n$.
With $n$ classes the \texttt{S-TRANS} rule can be applied a maximum of $n$ times.
Each time the \texttt{S-CLASS} rules is applied which sets in the variables $\ol{X}$ into the supertype.
This operations also runs in polynomial time, so the subtyping relation is decidable in polynomial time.

This shows that the time complexity of the GFJ type check is at least
polynomial or better. 
\hfill $\square$

\end{document}